\documentclass[a4paper,UKenglish,cleveref,autoref]{lipics-v2019}

\graphicspath{{./graphics/}}

\title{Retracting Graphs to Cycles}

\author{Samuel Haney}{Duke University} {shaney@cs.duke.edu}{[orcid]}{}

\author{Mehraneh Liaee}{Northeastern University} {mehraneh@ccs.neu.edu}{[orcid]}{}

\author{Bruce M. Maggs}{Duke University \& Akamai Technologies} {bmm@cs.duke.edu}{[orcid]}{}

\author{Debmalya Panigrahi}{Duke University} {debmalya@cs.duke.edu}{[orcid]}{}

\author{Rajmohan Rajaraman}{Northeastern University} {rraj@ccs.neu.edu}{[orcid]}{}

\author{Ravi Sundaram}{Northeastern University} {koods@ccs.neu.edu}{[orcid]}{}

\authorrunning{S. Haney, M. Liaee, B. M. Maggs, D. Panigrahi, R. Rajaraman, R. Sundaram}

\Copyright{Samuel Haney, Mehraneh Liaee, Bruce M. Maggs, Debmalya Panigrahi, Rajmohan Rajaraman, Ravi Sundaram}



\category{Regular Paper Track A}

\funding{Samuel Haney and Bruce M. Maggs were supported by NSF grant CCF 1535972, Debmalya Panigrahi was partially supported by NSF grants CCF 1527084 and CCF 1535972, an NSF CAREER Award CCF 1750140, and a Joint Indo-US Networked Center for "Algorithms under Uncertainty", Ravi Sundaram was supported by NSF grants CCF 1535929 and NSF CNS 1718286.}

\acknowledgements{We would like to thank Joseph Seffi Naor for helpful discussions on the problems considered in this paper.}

\nolinenumbers 

\hideLIPIcs  

\usepackage{mathtools,amsmath,amssymb}
\usepackage{graphicx}
\usepackage[colorinlistoftodos]{todonotes}

\usepackage{algorithm}
\usepackage{algpseudocode}
\algnewcommand\algorithmicinput{\textbf{Input:}}
\algnewcommand\algorithmicoutput{\textbf{Output:}}
\algnewcommand\Input{\item[\algorithmicinput]}%
\algnewcommand\Output{\item[\algorithmicoutput]}%

\usepackage{amsthm}             
\usepackage{xspace}             
\usepackage{amsfonts}             
\usepackage{cleveref}           
\usepackage{enumerate}
\usepackage{subcaption}

\newcommand{\junk}[1]{}
\newcommand{\eat}[1]{}
\newcommand{\BfPara}[1]{\noindent {\bf #1.}}

\EventEditors{Christel Baier, Ioannis Chatzigiannakis, Paola Flocchini, and Stefano Leonardi}
\EventNoEds{4}
\EventLongTitle{46th International Colloquium on Automata, Languages, and Programming (ICALP 2019)}
\EventShortTitle{ICALP 2019}
\EventAcronym{ICALP}
\EventYear{2019}
\EventDate{July 9--12, 2019}
\EventLocation{Patras, Greece}
\EventLogo{eatcs}
\SeriesVolume{132}
\ArticleNo{65}

\begin{document}
\maketitle

\begin{abstract}
We initiate the algorithmic study of retracting a graph into a cycle in the graph, which seeks a mapping of the graph vertices to the cycle vertices, so as to minimize the maximum stretch of any edge, subject to the constraint that the restriction of the mapping to the cycle is the identity map. This problem has its roots in the rich theory of retraction of topological spaces, and has strong ties to well-studied metric embedding problems such as minimum bandwidth and $0$-extension.
\eat{
The concept of a retraction, originally from topology, has led to the development of rich theories in both topology and graph theory. In this paper we introduce and study the problem of finding a minimum stretch retract of a graph, from an algorithmic standpoint.

A {\it retraction} is a mapping from the vertices of a graph to given subgraph such that its restriction to the subgraph is the identity map. A retraction is said to be a {\it stretch-$k$} retraction if the {\it stretch} of the retraction, or the maximum distance between the images of the endpoints of any edge of the graph, is $k$, as measured in the subgraph. To the best of our knowledge the problem of finding a minimum stretch retraction has not been explicitly studied, though it is closely related to a host of problems in the general area of metric embeddings including well-known ones such as the $0$-extension problem and the bandwidth problem. 
}
Our first result is an $O(\min\{k, \sqrt{n}\})$-approximation for retracting any graph on $n$ nodes to a cycle with $k$ nodes. We also show a surprising connection to Sperner's Lemma that rules out the possibility of improving this result using certain natural convex relaxations of the problem. Nevertheless, if the problem is restricted to planar graphs, we show that we can overcome these integrality gaps using an optimal combinatorial algorithm, which is the technical centerpiece of the paper.  Building on our planar graph algorithm, we also obtain a constant-factor approximation algorithm for retraction of points in the Euclidean plane to a uniform cycle. 

\end{abstract}


\newpage
\section{Introduction}
Originally introduced in 1930 by K. Borsuk in his PhD thesis~\cite{Borsuk1930}, {\em retraction} is a fundamental concept in topology describing continuous mappings of a topological space into a subspace that leaves the position of all points in the subspace fixed. Over the years, this has developed into a rich theory with deep connections to fundamental results in topology such as Brouwer's Fixed Point Theorem \cite{hocking1961topology}.
Inspired by this success, graph theorists have extensively studied a discrete version of the problem in graphs, where a {\it retraction} is a mapping from the vertices of a graph to a given subgraph that produces the identity map when restricted to the subgraph (i.e., it leaves the subgraph fixed).  For a rich history of retraction in graph theory, we refer the reader to~\cite{hell2004graphs}.  Define the {\it stretch} of a retraction to be the maximum distance between the images of the endpoints of any edge, as measured in the subgraph. We use stretch-$k$ retraction to mean a retraction whose stretch is $k$; in particular, a stretch-$1$ retraction is a mapping where every edge of the graph is mapped to either an edge of the subgraph, or both its ends are mapped to the same vertex of the subgraph\footnote{In the literature, a stretch-1 retraction is often simply referred to as a retraction or a retract~\cite{hell2004graphs}.  Also, in many studies, a (stretch-1) retraction requires that the two end-points of an edge in the graph are mapped to two end-points of an edge in the subgraph.  These studies differentiate between the case where the subgraph being retracted to is reflexive (has self-loops) or irreflexive (no self-loops).  In this sense, our notion of graph retraction corresponds to their notion of retraction to a reflexive subgraph.}. 

In this paper, we study the algorithmic problem of finding a {\em minimum stretch retraction} in a graph.
This problem belongs to the rich area of metric embeddings, but somewhat surprisingly, has not received much attention in spite of the deep but non-constructive results in the graph theory literature. The graph retraction problem has a close resemblance to the well-studied $0$-extension problem~\cite{calinescu+kr:0-extension,karloff+kmr:0-extension,karzanov:0-extension} (and its generalizations such as metric labeling~\cite{kleinberg+t:label,chuzhoy+n:label}), which is also an embedding of a graph $G$ to a metric over a subset of terminals $H$ with the constraint that each vertex in $H$ maps to itself. The two problems differ in their objective: whereas $0$-extension seeks to minimize the {\em average} stretch of edges, graph retraction minimizes the {\em maximum} stretch. The different objectives lead to significant technical differences. For instance, a well-studied linear program called the earthmover LP has a nearly logarithmic integrality gap for $0$-extension.  In contrast, we show that a corresponding earthmover LP for graph retraction has integrality gap $\Omega(\sqrt{n})$. A well-studied problem in the metric embedding literature that considers the maximum stretch objective is the {\em minimum bandwidth} problem, where one seeks an isomorphic embedding of a graph into a line (or cycle) that minimizes maximum stretch. In contrast, in graph retraction, we allow homomorphic maps\footnote{A {\em homomorphic} map is one where an image can have multiple pre-images, while an {\em isomorphic} map requires that every image has at most one pre-image.} but additionally require a subset of vertices (the anchors) to be mapped to themselves.

From an applications standpoint, our original motivation for studying minimum-stretch graph retraction comes from a distributed systems scenario where the aim is to map processes comprising a distributed computation to a network of servers where some processes are constrained to be mapped onto specific servers. The objective is to minimize the maximum communication latency between two communicating processes in the embedding.  Such anchored embedding problems can be shown to be equivalent to graph retraction for general subgraphs, and arise in several other domains including VLSI layout, multi-processor placement, graph drawing, and visualization~\cite{held+krv:vlsi,hansen:embed,rotter+v:arrange}.

\subsection{Problem definition, techniques, and results}

We begin with a formal definition of the minimum stretch retraction problem.

\begin{definition}
Given an unweighted {\em guest}\/ graph $G = (V,E)$ and a {\em host}\/ subgraph $H = (A, E')$ of $G$, a mapping $f: V \rightarrow A$ is a retraction of $G$ to $H$ if $f(v) = v$ for all $v \in A$.  For a given retraction $f$ of $G$ to $H$, define the {\em stretch}\/ of an edge $e = (u,v) \in E(G)$ to be $d_H(f(u), f(v))$, where $d_H$ is the distance metric induced by $H$, and define the {\em stretch}\/ of $f$ to be the maximum stretch over all edges of graph $G$.  The goal of the {\em minimum-stretch graph retraction}\/ problem is to find a retraction of $G$ to $H$ with minimum stretch.  We refer to the vertices of $A$ as {\em anchors}. 
\end{definition}

The graph retraction problem is easy if the subgraph $H$ is acyclic (see, e.g., \cite{nowakowski+r:retract}); therefore, the first non-trivial problem is to retract a graph into a cycle. Indeed, this problem is NP-hard even when $H$ is just a 4-cycle~\cite{FEDER1998236}.  Given this intractability result, a natural goal is to obtain an algorithm for retracting graphs to cycles that {\em approximately} minimizes the stretch of the retraction.  This problem is the focus of our work.  While there has been considerable interest in identifying conditions under which retracting to a cycle with stretch 1 is tractable~\cite{Hahn1997,hell2004graphs,VIKAS2005406}, there has been no work (to the best of our knowledge) on deriving approximations to the minimum stretch\footnote{One direct implication of the NP-hardness proof is that approximating the maximum stretch to a multiplicative factor better than two is also NP-hard.}.

We consider the following lower bound for the problem: if anchors $u$ and $v$ are distance $\ell$ in $H$, and there exists a path of $p$ vertices in $G$ between $u$ and $v$, then every retraction has stretch at least $\ell / p$.
This lower bound turns out to be tight when $H$ is acyclic, which is the reason retraction to acyclic graphs is an easy problem.
However, this lower bound is no longer tight when $H$ is a cycle.
For example, consider a grid graph where $H$ is the border of the grid.
The lower bound given above says that any retraction has stretch at least $\Omega(1)$.
However, using the well-known Sperner's lemma, we show that the optimal retraction has stretch at least $\Omega(\sqrt{n})$.

Using just the simple distance based lower bound, we show that the gap on the grid is in fact the worst possible by giving a $O(\min\{k, \sqrt{n}\})$-approximation for the problem, where $k$ is the number of vertices of $H$.
Our algorithm works by first mapping vertices of the graph into a grid, then projecting vertices outward to the border 
from the largest {\em hole} in the grid, which is the largest region containing no vertices.  

\begin{theorem}
There is a deterministic, polynomial-time algorithm that computes a retraction of a graph to a cycle with stretch at most $\min\{k/2, O(\sqrt{n})\}$ times the optimal stretch, where $n$ and $k$ are respectively number of  vertices in the graph and cycle. 
\end{theorem}

Our results for retracting a general graph to a cycle appear in Section~\ref{sec:general}.
We also give evidence that the gap induced by Sperner's lemma on a grid graph is fundamental, showing an $\Omega(\min\{k, \sqrt{n}\})$ integrality gap for natural linear and semi-definite programming relaxations of the problem. To overcome this gap, we focus on the special case of planar graphs, of which the grid is an example. Retraction in planar graphs has been considered in the past, most notably in a beautiful paper of Quilliot \cite{QUILLIOT198561} who uses homotopy theory to characterize stretch-1 retractions of a planar graph to a cycle. Quillot's proof, however, does not yield an efficient algorithm.
In Section~\ref{sec:planar-graph-to-cycle}, we provide an exact algorithm for retraction in planar graphs by developing the gap induced by Sperner's lemma on a grid into a general lower bound on the optimal stretch for planar graphs. 

\begin{theorem}
\label{thm:planar-intro}
  There is a deterministic, polynomial-time algorithm that computes a retraction of a planar graph to a cycle with optimal stretch.
\end{theorem}

Unfortunately, our techniques rely heavily on the planarity of the graph, and do not appear to generalize to arbitrary graphs. While we leave the question of obtaining a better approximation for general graphs open, we provide a more sophisticated linear programming formulation that captures the Sperner lower bound on general graphs as a possible route to attack the problem. 

We also study natural special cases and generalizations of the problem, all of which are presented in the appendix due to space limitations.  First, we consider a geometric setting, where a set of points in the Euclidean plane has to be retracted to a uniform cycle of anchors. By a uniform cycle of anchors we mean a set of anchors which are distributed uniformly on a circle in the plane. We obtain a constant approximation algorithm for this problem, by building on our planar graph algorithm, in Appendix~\ref{sec:2d_euclidean}.  We next consider retraction of a graph of bounded treewidth to an {\em arbitrary subgraph}, and obtain a polynomial-time exact algorithm in Appendix~\ref{sec:treewidth}.  Finally, we apply the lower bound argument of~\cite{karloff+kmr:0-extension} for $0$-extension to show in Appendix~\ref{app:hardness} that a general variant of the problem that seeks a retraction of an arbitrary weighted graph $G$ to a metric over a subset of the vertices of $G$ is hard to approximate to within a factor of $\Omega(\log^{1/4 - \epsilon} n)$ for any $\epsilon > 0$.   

\subsection{Related work}

\BfPara{List homomorphisms and constraint satisfaction}
The graph retraction problem is a special case of the {\em list homomorphism} problem introduced by Feder and Hell~\cite{FEDER1998236}, who established conditions under which the problem is NP-complete. 
Given graphs $G, H$, and $L(v) \subset V(H)$ for each $v \in V(G)$, a list homomorphism of $G$ to $H$ with respect to $L$ is a homomorphism $f: G \rightarrow H$ with $f(v) \in L(v)$ for each $v \in V(G)$. 

Several special cases of graph retraction and variants of list homomorphism have been subsequently studied (e.g.,~\cite{feder+hjkn:retract,hell2004graphs,vikas:retract,VIKAS2005406}).  These studies have
established and exploited the rich connections between list homomorphism and  Constraint Satisfaction Problems (CSPs).   Though approximation algorithms for CSPs and related problems such as Label Cover have been extensively studied, the objective pursued there is that of maximizing the number of constraints that are satisfied.  For our graph retraction problem, this would correspond to maximizing the number of edges that have stretch below a certain threshold.  Our notion of approximation in graph retraction, however, is the least factor by which the stretch constraints need to be relaxed so that all edges are satisfied.

\smallskip
\BfPara{$0$-extension, minimum bandwidth, and low-distortion embeddings}
From an approximation algorithms standpoint, the graph retraction problem is closely related to the $0$-extension and minimum bandwidth problems~\cite{FEIGE2000510,BLUM200025,gupta:bandwidth,vempala:projectVLSI,dubey+fu:bandwidth,Saxe1980DynamicProgrammingAF}. In the $0$-extension problem, one seeks to minimize the average stretch, which can be solved to an $O(\log k/\log\log k)$ approximation using a natural LP relaxation~\cite{calinescu+kr:0-extension,fakcharoenphol+hrt:0-extension}. In contrast, we give polynomial integrality gaps for the graph retraction problem. In the minimum bandwidth problem, the objective is to find an embedding to a line that minimizes maximum stretch, but the constraint is that the map must be isomorphic  rather than that the anchor vertices must be fixed. In a seminal result~\cite{FEIGE2000510}, Feige designed the first polylogarithmic-approximation using a novel concept of volume-respecting embeddings.  A slightly improved approximation was achieved in~\cite{dunagan+v:bandwidth} by combining Feige's approach with another bandwidth algorithm based on semidefinite-programming~\cite{BLUM200025}.  Interestingly, the minimum bandwidth problem is NP-hard even for (guest) trees, while graph retraction to (host) trees is solvable in polynomial time. Conversely, the bandwidth problem is solvable in time $O(n^b)$ for bandwidth $b$ graphs \cite{GurariS84}, while graph retraction to a cycle is NP-complete even when the host cycle has only four vertices.  Nevertheless, it is conceivable that volume-respecting embeddings, in combination with random projection, could lead to effective approximation algorithms for graph retraction to a cycle in a manner similar to what was achieved for VLSI layout on the plane~\cite{vempala:projectVLSI}.  

Also related are the well-studied variants of linear and circular arrangements, but their objective functions are average stretch, as opposed to maximum stretch. Finally, another related area is that of {\em low-distortion embeddings} (e.g.,~\cite{matousek:embed}), where recent work has considered embedding one specific $n$-point metric to another $n$-point metric~\cite{kenyon+rs:embed,matousek+s:embed,badoiu+dgrrrs:embed} similar to the graph retraction problem. But low-distortion embeddings typically require {\em non-contracting}\/ isomorphic maps, which distinguishes them significantly from the graph retraction problem.  

A related recent work studies low-distortion {\em contractions\/} of graphs~\cite{bernstein+ddkms:contract}.  Specifically, the goal is to determine a maximum number of edge contractions of a given graph $G$ such that for every pair of vertices, the distance between corresponding vertices in the contracted graph is at least a given affine function of the distance in $G$.  Several upper bounds and hardness of approximations are presented in~\cite{bernstein+ddkms:contract} for many special cases and problem variants.  While graph retraction and contraction problems share the notion of mapping to a subgraph, the problems are considerably different; for instance, in the graph retraction problem the subgraph $H$ is part of the input, and the objective is to minimize the maximum stretch.

\eat{
\smallskip
\BfPara{Low-distortion embeddings}
The technique of volume-respecting embeddings~\cite{FEIGE2000510,krauthgamer+lmn:embed} is one instance of the general paradigm of low-distortion metric embeddings that has been extensively used in solving diverse graph optimization problems~\cite{matousek:embed}.  While the research focus in metric embeddings is primarily on embedding one class of metrics in another class (e.g., general metrics to trees, general metrics to Euclidean, $\ell_2$ to $\ell_1$), recent research has also concerned the problem of embedding one specific $n$-point metric to another $n$-point metric~\cite{kenyon+rs:embed,matousek+s:embed}.  The embeddings that are sought here are one-to-one, and for most problems of interest, the distortion bounds as well as approximation factors likely achievable are polynomial in $n$~\cite{matousek+s:embed,badoiu+dgrrrs:embed}.  One significant technical difference between low-distortion embeddings and graph retraction is that the embeddings being sought in the former are {\em non-contracting}\/ while no such restriction is present for the retraction or the bandwidth problem, as a result of which approximation factors do not easily transfer from one problem to the other~\cite{badoiu+dgrrrs:embed}.
}



\section{Retracting an arbitrary graph to a cycle}
\label{sec:general}
In this section, we study the problem of retracting an arbitrary graph to a cycle over a subset of vertices of the graph.  Let $G$ denote the guest graph over a set $V$ of $n$ vertices, with shortest path distance function $d_G$.  Let $H$ denote the host cycle with shortest path distance function $d_H$ over a subset $A \subseteq V$ of $k$ anchors.  

Arguably, the simplest lower bound on the optimal stretch is the {\em distance-based} bound $\ell(G,H) = \max_{u,v \in A} d_H(u,v)/d_G(u,v)$, since every retraction places a path of length $d_G(u,v)$ in $G$ on a path of length at least $d_H(u,v)$ in $H$. 

We now present our algorithm (Algorithm~\ref{alg:grid_embedding}), which achieves a stretch of $\min\{k/2, \ell(G,H)\sqrt{n}\}$. Here, we give a high level overview of the algorithm.  The first step of algorithm is to embed the input graph $G$ into a grid of size $k/4 \times k/4$ subject to some constraints. The second step is to find the largest empty sub-grid $D$ such that no point is mapped inside of $D$ and center of $D$ is within a desirable distance from center of grid $M$. And final step is to project the points in grid $M$ to its boundary with respect to center of sub-grid $D$.  

\junk{Let $S$ denote a disk of unit radius and let $V$ denote a set of $n$ points in the disk, $k$ of which are uniformly placed on the boundary of the disk.  Let $D$ denote the largest ball inside $S$ that does not contain any point in $V$. Let $r$ denote the radius of $D$.  For any point $v$, let $R(v)$ denote the ray originating from the center of $D$ passing through $v$.  We define the {\em projection embedding}\/ $\pi$ as follows: for each point $v$, $\pi(v)$ is the anchor nearest in the clockwise direction to the intersection of $R(v)$ with the boundary of the disk.  
}

\begin{algorithm}
\caption{Algorithm for retracting an arbitrary graph to a cycle \label{alg:grid_embedding}}
\begin{algorithmic}
\Input {Graph $G$, host cycle $H$}
\Output {Embedding function $f$}
\State {{\bf Embedding in a grid:} Determine embedding $g$ from $G$ into a $k/4 \times k/4$ grid $M$ such that $H$ is embedded one-to-one to the boundary of $M$ and for every $u, v \in V$, $d_\infty(g(u),g(v)) \le \ell(G,H) d_G(u,v)$.}
\State {{\bf Find largest hole:} Find the largest square sub-grid $D$ of $M$ such that (a) its center $c$ is at $L_\infty$ distance at most $k/16$ from the center of $M$ and (b) there is no vertex $u$ in $G$ for which $g(u)$ is in the interior of $D$.}
\State{{\bf Projection embedding:} For all $v$ in $G$:
\begin{enumerate}
\item $R(v) \leftarrow$ ray originating from the center of $D$ and passing through $g(v)$.
\item $f(v) \leftarrow$ the anchor on the boundary of grid $M$ nearest in the clockwise direction to the intersection of $R(v)$ with the boundary of $M$.
\end{enumerate}
}
\junk{\State {Use the projection embedding of Section~\ref{sec:2d_euclidean.disk} to 
map each point $g(u)$ to a vertex $f(u)$ in the boundary $H$.}} \\
\Return $f$
\end{algorithmic}
\end{algorithm}
We now show how to implement the first step of Algorithm~\ref{alg:grid_embedding}.  Our goal is to embed each vertex \(u \in G\) to some point \(g(u)\) in a \(k/4 \times k/4\) grid such that for every \(u,v\), we have the following inequality, where $d_\infty(a,b)$ denotes the $L_\infty$ distance between $a$ and $b$. (That is, for two points $(x_1,y_1)$ and $ (x_2,y_2)$, $ d_\infty((x_1,y_1), (x_2,y_2)) = \max\{|x_1-x_2|, |y_1 - y_2|\}$.)
\begin{equation}
\label{eqn:grid_embedding}
  d_{\infty}(g(u),g(v)) \le \ell(G,H) d_G(u,v)
\end{equation}
Additionally, we require that \(H\) is embedded to the boundary of the grid, such that adjacent anchors lie on adjacent grid points.

\begin{lemma}
\label{lem:grid_embedding}
For every \(G\), we can find an embedding $g$ satisfying inequality~\ref{eqn:grid_embedding}.
\end{lemma}
\begin{proof}
We incrementally construct the embedding $g$.  Initially, we place the anchors on the boundary of the grid so that the boundary is isometric to $d_H$.  (This can be done since $H$ is a cycle.)  Since $d_\infty(g(u),g(v)) \le d_H(u,v)$ and $d_H(u,v) \le \ell(G,H) d_G(u,v)$,  inequality~\ref{eqn:grid_embedding} holds for all anchors $u$ and $v$ in $H$.

We next inductively embed the remaining vertices of $G$.  Suppose we need to embed vertex \(v_i\), and vertices \(U = v_1,\ldots,v_{i-1}\) have already been embedded.  Assume inductively that the embedding of the vertices of \(U\) satisfies inequality~\ref{eqn:grid_embedding} for the vertices in $U$.  

Let \(B_\infty(g(u),r)\) denote the \(L_\infty\) ball around \(g(u)\) with radius \(r\) (note that these balls are axis-aligned squares).  Let $x$ be any point in \(\bigcap_{u \in U} B_\infty(g(u),\ell(G,H) d_G(u,v_i))\).  If we set $g(v_i) = x$, then inequality~\ref{eqn:grid_embedding} holds for all points in $U \cup \{v_i\}$.  We now show that this intersection is nonempty (it is straightforward to find an element in the intersection).  The set of axis aligned squares has Helly number 2\footnote{A family of sets has Helly number $h$ if any minimal subfamily with an empty intersection has $h$ or fewer sets in it.}; therefore it is enough to show that for every \(u,u' \in U\), \(B_\infty(g(u),\ell(G,H)d_G(u,v_i))\) and \(B_\infty(g(u'),\ell(G,H)d_G(u',v_i))\) intersect.  Otherwise,
\begin{align*}
  d_\infty(g(u),g(u')) > \ell(G,H)(d_G(u,v_i) + d_G(u',v_i)) \ge \ell(G,H)d_G(u,u').
\end{align*}
This contradicts our induction hypothesis that the set of vertices in \(U\) satisfies inequality~\ref{eqn:grid_embedding}, and completes the proof of the lemma.
\end{proof}
In the following lemma, we analyze the projection embedding step of the algorithm.
\begin{figure}[h]
  \centering
  \begin{subfigure}[c]{.3\textwidth}
    \includegraphics[width=\textwidth,page=1]{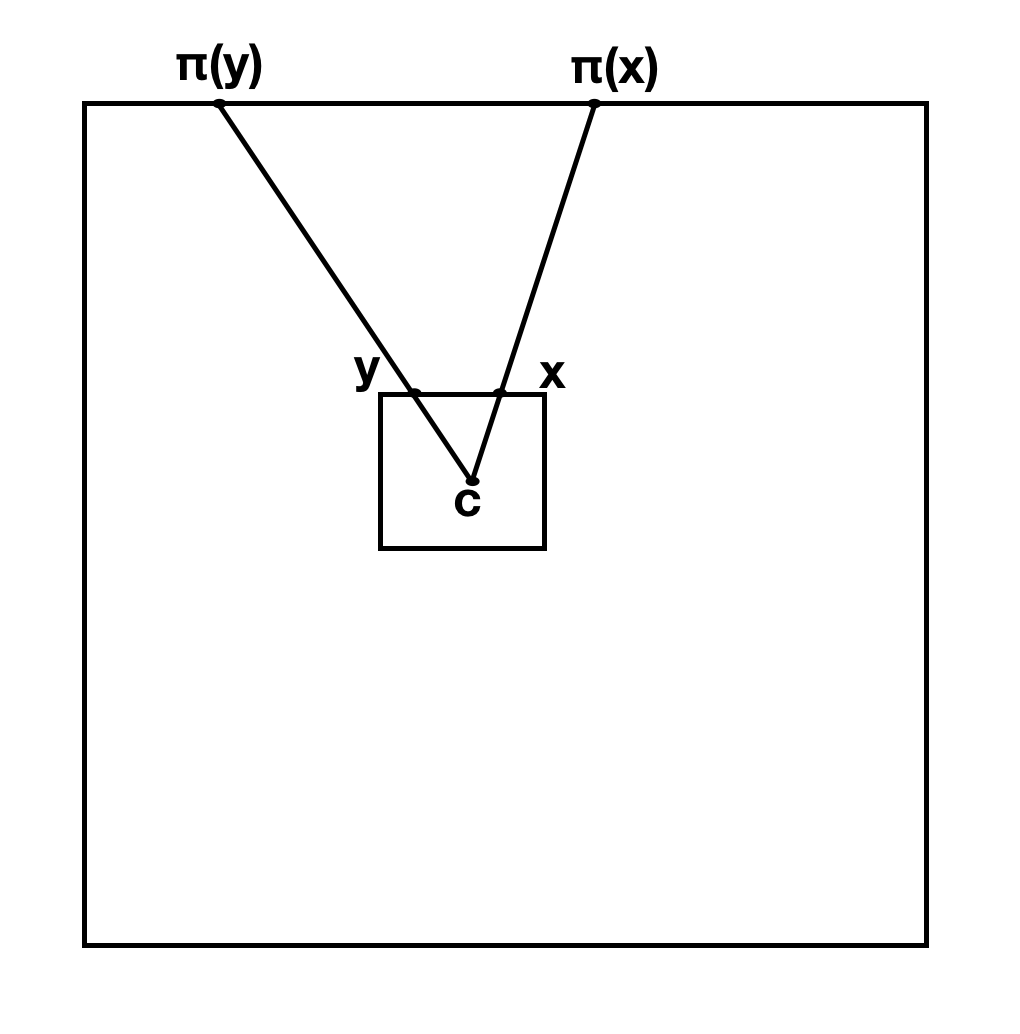}
    \caption{Points $x$ and $y$ are on the same side of square $D$, and points $\pi(x)$ and $\pi(y)$ are on one side of boundary of $M$ parallel to segment $\overline{xy}$.}
\label{fig:square_embedding_a}
\end{subfigure}
  \begin{subfigure}[c]{.3\textwidth}
   \includegraphics[width=1.05\textwidth,page=2]{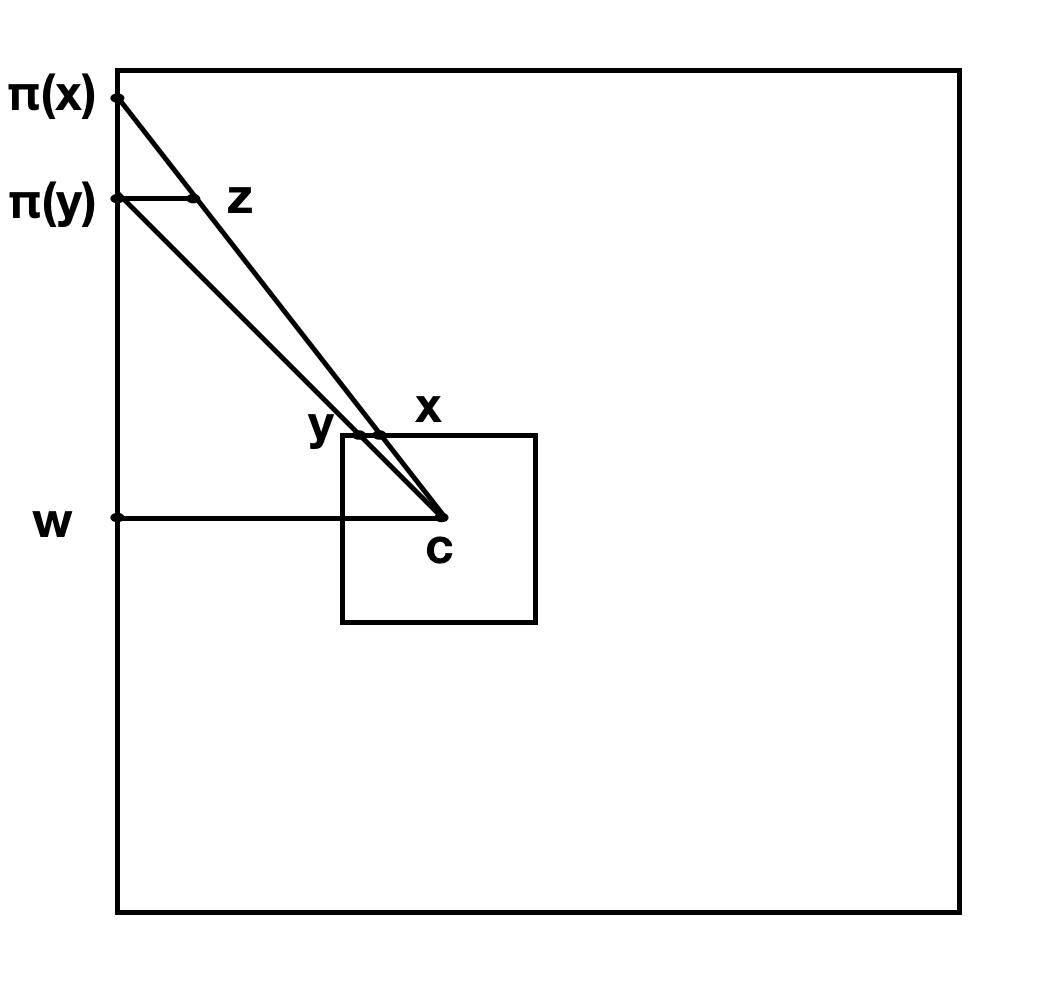}
    \caption{Points $x$ and $y$ are on the same side of square $D$, and points $\pi(x)$ and $\pi(y)$ are on one side of boundary of M orthogonal to segment $\overline{xy}$.}
    \label{fig:square_embedding_b}
  \end{subfigure}
  \begin{subfigure}[c]{.3\textwidth}
    \includegraphics[width=1.05\textwidth]{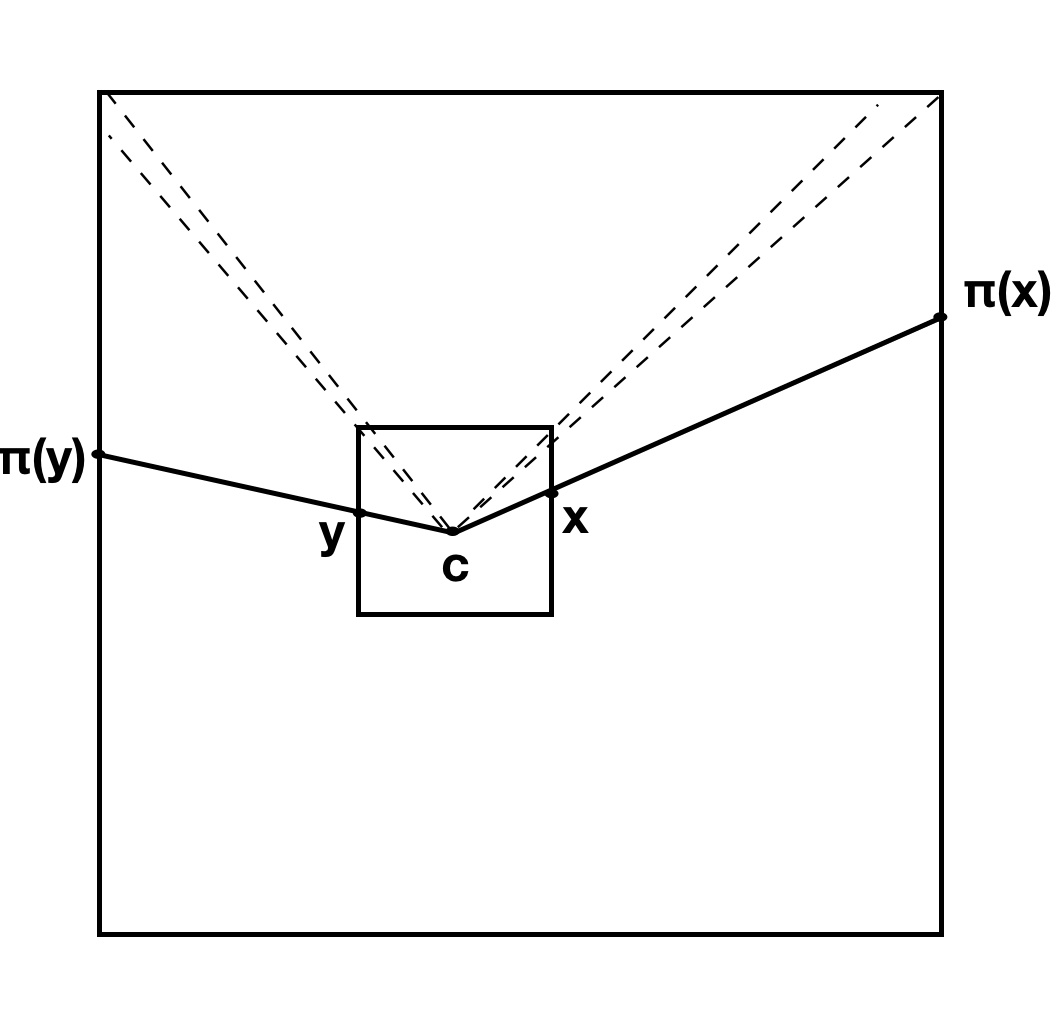}
    \caption{General case where Points $x$ and $y$ (resp. points $\pi(x)$ and $\pi(y)$) are anywhere on the boundary of $D$ (resp. on the boundary of $M$)}
    \label{fig:square_embedding_c}
  \end{subfigure}
  \caption{Embedding of points inside the grid $M$ to its boundary using an empty square $D$.  Referred to in the proof of Lemma~\ref{disk_embedding_theorem}. }
  \label{fig:square_embedding}
\end{figure}
\begin{lemma} \label{disk_embedding_theorem}
Suppose $r$ is the side length of the largest empty square $D$ inside $M$.  Then for any vertices $u$ and $v$ in $G$,  $d_H(f(u),f(v))$ is at most $1 + (10\sqrt{2}k/r)d_\infty(g(u),g(v))$.    
\end{lemma}
\begin{proof} 
For any point $x$, let $\pi(x)$ denote the intersection of the boundary of $M$ and the ray from the center $c$ of $D$ passing through $x$.  Note that for any vertex $v$ in $G$, $f(v)$ is the anchor in $H$ nearest in clockwise direction to $\pi(g(v))$.  We show that for any $x, y \in M$, the distance between $\pi(x)$ and $\pi(y)$ along the boundary of $M$ is at most $(10\sqrt{2}k/r)d_\infty(x,y)$.  

We first argue that it is sufficient to establish the preceding claim for points on the boundary of $D$, at the loss of a factor of $\sqrt{2}$.  Let $x$ and $y$ be two arbitrary points in $M$ but not in the interior of $D$.  Let $x'$ (resp., $y'$) denote the intersection of $R(x)$ (resp., $R(y)$) and the boundary of $D$.  From elementary geometry, it follows that $d(x',y') \le d(x,y)$, where $d$ is the Euclidean distance; since $d_\infty(x,y) \ge d(x,y)/\sqrt{2}$ and $d_\infty(x',y') \le d(x',y')$, we obtain $d_\infty(x',y') \le \sqrt{2}d_\infty(x,y)$.  Since $\pi(x) = \pi(x')$ and $\pi(y) = \pi(y')$, establishing the above statement for $x'$ and $y'$ implies the same for $x$ and $y$, up to a factor of $\sqrt{2}$.

Consider points $x$ and $y$ on the boundary of $D$.  We consider three cases.  In the first two cases, $x$ and $y$ are on the same side of $D$. In the first case (Figure~\ref{fig:square_embedding_a}), $\pi(x)$ and $\pi(y)$ are on the same side of the boundary of $M$ and segment $\overline{\pi(x)\pi(y)}$ is parallel to segment $\overline{xy}$; then, by similarity of triangle formed by $c$, $x$, and $y$ and the one formed by $c$, $\pi(x)$ and $\pi(y)$, we obtain that the distance between $\pi(x)$ and $\pi(y)$ is at most $3kd_\infty(x,y)/(16r)$.  In the second case (Figure~\ref{fig:square_embedding_b}), $\pi(x)$ and $\pi(y)$ are on same side of the boundary of $M$, and segment $\overline{\pi(x)\pi(y)}$ is orthogonal to segment $\overline{xy}$. In this case, w.l.o.g. assume that $\pi(y)$ is closer to center $c$ than $\pi(x)$ with respect to $d_{\infty}$ distance. Let point $z$ be a point on segment $\overline{c\pi(x)}$ such that segments $\overline{xy}$ and $\overline{\pi(y)z}$ are parallel. From center $c$ extend a line parallel to segment $\overline{xy}$ until it hits the side of $M$ on which $\pi(x)$ and $\pi(y)$ are. Let $w$ be the intersection. Using elementary geometry and similarity argument, we have the following: 
\begin{eqnarray*}
\frac{|\overline{\pi(x)\pi(y)}|}{|\overline{z\pi(y)}|} = \frac{|\overline{\pi(x)w|}}{|\overline{cw}|}\leq \frac{k/4}{k/16} = 4 \;\;\;\; \mbox{ and } \frac{\overline{z\pi(y)}}{\overline{xy}} = \frac{\overline{\pi(y)w}}{r} \leq \frac{k}{4r}
\end{eqnarray*}
We thus obtain $\frac{|\overline{\pi(x)\pi(y)}|}{|\overline{xy}|} \leq k/r$.  For the third case (Figure~\ref{fig:square_embedding_c}), we observe that $d_\infty(x,y)$ is at least half the shortest path between $x$ and $y$ that lies within the boundary of $D$.  This latter shortest path consists of at most five segments, each residing completely on one side of the boundary of $D$.  We apply the argument of the first and second case to each of these segments to obtain that the distance between $\pi(x)$ and $\pi(y)$ is at most $10kd_\infty(x,y)/r$.

To complete the proof, we note that distance between anchor nearest (clockwise) to $\pi(x)$ and anchor nearest (clockwise) to $\pi(y)$ is at most one plus the distance between $\pi(x)$ and $\pi(y)$.  Therefore, the $d_H(f(u),f(v))$ is at most $1 + 10\sqrt{2}kd_\infty(g(u),g(v))/r$.
\junk{Let $A$ and $B$ denote their mapped points on the boundary of the disk, $C$ denote the center of ball $D$, $O$ denote the center of the disk, and  $\alpha$ denote the ratio $\frac{dist(A,B)}{dist(X,Y)}$. 
By $dist(A,B)$, we mean length of segment $AB$. 

We need to show that $\alpha = O(\frac{1}{r})$. 
Let $\theta$ be the angle $\widehat{AOB}$, and $\theta'$ be the angle $\widehat{XCY}$. Then the length of arc $AB$ is $\theta$, and we know that $dist(A,B) \leq \theta$. Also the length of arc $XY$ is $r\theta'$, and it is easy to verify that for any $0 < \theta' \leq \pi$, we have $\frac{2}{\pi} r \theta' \leq dist(X,Y)$. Extend segment $AC$ from $C$ and let $C'$ be the point that line $AC$ hits boundary of disk. It is easy to verify that $\widehat{AC'B} = \frac{\theta}{2}$, and $\widehat{ACB} > \widehat{AC'B}$. Thus, $\theta' \geq \frac{\theta}{2}$. Now, we have 
\[
\alpha = \frac{dist(A,B)}{dist(X,Y)} \leq \frac{\theta}{\frac{2}{\pi}r\theta'} \leq \frac{\pi}{r} 
\]
This completes the proof that any pair of points on the boundary of ball $D$ will get stretched by a factor $O(\frac{1}{r})$. }
\end{proof}
\begin{theorem}
\label{thm:gen-retraction}
Algorithm~\ref{alg:grid_embedding} computes a retraction of $G$ to the cycle $H$ with stretch at most the minimum of $k/2$ and $O(\sqrt{n})$ times the optimal stretch.
\end{theorem}
\begin{proof}
By Lemma~\ref{lem:grid_embedding}, the embedding $g$ satisfies inequality~\ref{eqn:grid_embedding} for every $u$ and $v$ in $G$.  By a straightforward averaging argument, there exists a square of side length $k/(8\sqrt{n})$ whose center is at $L_\infty$ distance at most $k/16$ from the center of $M$ and which does not contain $g(u)$ for any $u$ in $V$.  By Lemma~\ref{disk_embedding_theorem}, the projection embedding ensures that for any $u$ and $v$ in $V$, $d_H(f(u),f(v))$ is at most $1 + O(\sqrt{n}) \ell(G,H) d_G(u,v)$.  Since the distance in $H$ cannot exceed $k/2$, the claim of the theorem follows.
\end{proof}

\junk{
We show that finding an embedding that leaves the largest empty square is equivalent to finding the minimum stretch retraction up to a factor of 2.

\begin{claim}
If there is a retraction of \(G\) to \(H\), then there is an embedding to the border of the grid.
Conversely, an embedding to the border of the grid implies a retraction with stretch at most 2.
\end{claim}

\begin{proof}
If there is a retraction, then mapping the points to the grid where they are mapped by the retraction is feasible.

In the other direction, suppose there is an embedding to the border of the grid.
The distance between any pair of points along the border of the grid is at most twice their distance in the \(L_\infty\) metric.
Therefore, this embedding is a stretch-2 retraction.
\end{proof}
}

\BfPara{The Sperner bottleneck}
Unfortunately, we cannot improve on the approximation ratio in Theorem~\ref{thm:gen-retraction} using only the distance-based lower bound. Consider the following instance: the guest graph $G$ is
the $\sqrt{n} \times \sqrt{n}$ grid, and the host $H$ is the cycle of $G$ formed by the $4\sqrt{n}$ vertices on the outer boundary of $G$. It is easy to see that the distance-based lower bound has a value of $2$ on this instance. On the other hand, using Sperner's Lemma from topology, we show that a stretch of $o(\sqrt{n})$ is ruled out:
\begin{lemma}
\label{lem:general.lower_bounds}
The optimal stretch achievable for an $n$-vertex grid is at least $2 \sqrt{n}/3$.
\end{lemma}
\begin{proof}
Suppose we triangulate the grid by adding northwest-to-southeast
diagonals in each cell of the grid.  Consider the following coloring
of the boundary $H$ with 3 colors.  Divide $H$ into three
segments, each consisting of a contiguous sequence of at least $\lfloor
4\sqrt{n}/3\rfloor$ vertices; all vertices in the first, second, and
third segment are colored red, green, and blue, respectively.  Let $f$
be any retraction from $G$ to $H$.  Let $c_f$ denote the following
coloring for $G \setminus H$: the color of $u$ is the color of $f(u)$.
By Sperner's Lemma~\cite{Sperner1928}, there exists a tri-chromatic
triangle.  This implies that there are two vertices within distance at
most two in $G$ that are at least $4\sqrt{n}/3$ apart in the
retraction $f$, resulting in a stretch of at least $2 \sqrt{n}/3$.
\end{proof}
Note that $k = \Theta(\sqrt{n})$ in this instance, so the above lemma
also rules out an $o(k)$ approximation using the distance-based lower
bound.  A natural approach to improving the approximation factor is to
use an LP or SDP relaxation for the problem. Indeed, the so-called
{\em earthmover LP} used for the closely related $0$-extension
problem~\cite{karloff+kmr:0-extension,chekuri+knz:label} can be easily
adapted to our minimum stretch retraction problem. Similarly, SDP
relaxations previously used for minimum bandwidth and related
problems~\cite{BLUM200025,manchoso+ye:localize} can also be adapted to
our problem. However, these convex relaxations also have an
integrality gap of $\Omega(\sqrt{n})$ for precisely the same reason:
they capture the distance-based lower bound but not the one from
Sperner's lemma on the grid (see Appendix~\ref{app:general} for a
detailed discussion of these LP/SDP relaxations and integrality
gaps).

In spite of these gaps, we show that the grid is not a particularly
challenging instance of the problem. In fact, in the next section, we
give an exact algorithm for retraction in planar graphs, of which
the grid is an example.  Retraction of planar graphs to cycles has
been considered in the past, and non-constructive characterizations of
stretch-$1$ embeddings were known~\cite{QUILLIOT198561}. Our
constructive result, while using planarity extensively, suggests that
there might be a general technique for addressing the Sperner
bottleneck described above. Indeed, we give a candidate LP relaxation
in Appendix~\ref{app:general} that captures the Sperner bound on the
grid. Rounding this LP to obtain a better approximation ratio, or
showing an integrality gap for it, is an interesting open question.


\section{Retracting a planar graph to a cycle}
\label{sec:planar-graph-to-cycle}
The main result of this section is the following theorem.
\begin{theorem}
\label{thm:planar}
  Let $G$ be a planar graph and $H$ a cycle of $G$.
  Then there is a polynomial time algorithm that finds a retraction from $G$ to $H$ with optimal stretch.
\end{theorem}
We begin by presenting some useful definitions and elementary claims in Section~\ref{sec:planar.prelim}.  We then present an overview of our algorithm in Section~\ref{sec:planar.overview}.  Finally, we present the algorithm and its analysis, leading to the proof of Theorem~\ref{thm:planar}.

\subsection{Preliminaries}
\label{sec:planar.prelim}
We begin with a simple lemma that reduces the problem of finding a minimum-stretch retraction to the problem of finding a stretch-1 retraction, in polynomial time.  Formally, suppose we have an algorithm \(\mathcal{A}\) that, given graphs \(G\) and \(H\) either finds a stretch-1 retraction from \(G\) to \(H\), or proves that no such retraction exists.  Then, we can use this algorithm to find the minimum stretch embedding of \(G\) into \(H\), using Lemma~\ref{lem:reduction-to-stretch-1} below, whose straightforward proof is deferred to Appendix~\ref{app:planar}.  Let $G_k$ be the graph where we replace each edge $e \in G, e \not\in H$ with a path of $k$ edges.  Clearly, $G_k$ can be computed in polynomial time.
\begin{lemma}
\label{lem:reduction-to-stretch-1}
\(G\) can be retracted to \(H\) with stretch \(k\) if and only if \(G_k\) can be retracted in \(H\) with stretch-1.
\end{lemma}
The following lemma, proved in Appendix~\ref{app:planar}, implies that degree-1 vertices can be eliminated.
\begin{lemma}
\label{lem:2-connected}
Without loss of generality, we can assume $G$ is 2-vertex connected.
\end{lemma}

Lemmas~\ref{lem:reduction-to-stretch-1} and~\ref{lem:2-connected} apply to general graphs.  In the rest of this subsection, we focus our attention on planar graphs.  We note that all the transformations in Lemmas~\ref{lem:reduction-to-stretch-1} and~\ref{lem:2-connected} preserve planarity of the graph.
In 2-connected planar graph, every face of a plane embedding is bordered by a simple cycle.
Finally, we can assume that there is a planar embedding of $G$ with $H$ bordering the outer face.
If this is not the case, $G \setminus H$ contains at least two connected components, which can each be retracted independently.

Next, we give some definitions related to planar graphs.
We call \(G\) \emph{triangulated} if it is maximally planar, i.e., adding any edge results in a graph that is not planar.
Equivalently, \(G\) is triangulated if every face of a plane embedding (including the outer face) of \(G\) has 3 edges.
We will make use of the Jordan curve theorem, which says that any closed loop partitions the plane into an inner and outer region (see e.g. \cite{armstrong2013basic}).
In particular, this implies that any curve crossing from the inner to the outer region intersects the loop.
For some cycle $C$ in $G$ and a plane embedding of $G$, we denote the subset of $\mathbb{R}^2$ surrounded by $C$ as $R_C$ (including the intersection with $C$ itself).
We say that $R \subset \mathbb{R}^2$ is \emph{inside} cycle \(C\) of $G$ for a plane embedding if \(R \subseteq R_C\).
If $R$ is inside $C$, we also say that $C$ {\em surrounds} $R$.
In a slight abuse of notation, we say $C$ surrounds subgraph $G'$ of $G$ for some fixed plane embedding, if $C$ surrounds the subset of $\mathbb{R}^2$ on which $G'$ is drawn in the plane embedding.

\subsection{Overview of our algorithm}
\label{sec:planar.overview}
Consider some plane embedding of graph $G$ such that $H$ is the
subgraph of $G$ bordering $G$'s outer face.  We give a polynomial-time
algorithm that finds a stretch-1 retraction from $G$ to $H$ or proves
that none exists.  Using Lemma~\ref{lem:reduction-to-stretch-1}, this
immediately yields an algorithm that finds a minimum stretch
retraction from $G$ to $H$.

Fix a planar embedding of $G$, let $H$ be defined as above, and let $F$ be a bounded face of $G$.  A key component of our algorithm is to find a suitable set of curves connecting $F$ to $H$.   Our aim is to find a set of $k = |V(H)|$ curves in $\mathbb{R}^2$ such that the following hold.
\begin{itemize}
\item Each curve begins at a distinct vertex of $F$ and ends at a distinct vertex of $H$.
\item The curves do not intersect each other.
\item A curve that intersects an edge of $G$ either contains the edge, or intersects the edge only at its vertices.
\item Each curve lies totally in $R_H \setminus F$.
\end{itemize}
We call curves with these properties \emph{valid} with respect to $F$.
We argue that the curves partition $R_H \setminus F$ (up to their boundaries being duplicated) into a set of regions.
Each of these regions is defined by the subset of $\mathbb{R}^2$ surrounded by the closed loop made up of two of the aforementioned curves, a single edge of $H$, and a path on the boundary of $F$.

\begin{figure}[t]
  \centering
  \begin{subfigure}[c]{.3\textwidth}
    \includegraphics[width=.8\textwidth,page=1]{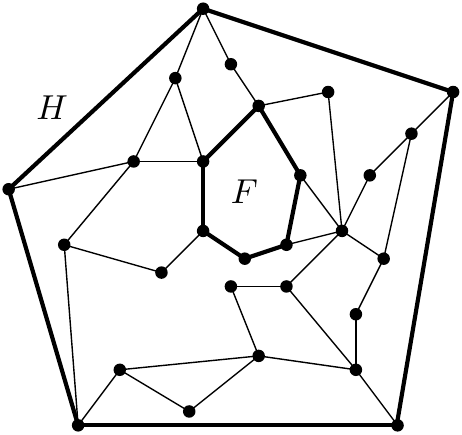}
    \caption{A graph $G$. The outer cycle $H$ and the face $F$ are shown in bold.}
  \end{subfigure}
  \begin{subfigure}[c]{.3\textwidth}
    \includegraphics[width=.8\textwidth,page=2]{curves-and-regions.pdf}
    \caption{Non-intersecting curves partition the region contained in $H$ but not in $F$.}
  \end{subfigure}
  \begin{subfigure}[c]{.3\textwidth}
    \includegraphics[width=.8\textwidth]{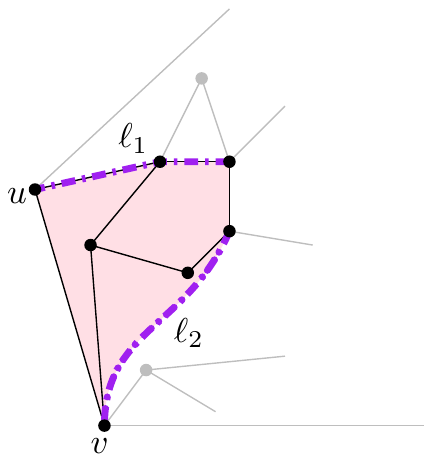}
    \caption{Vertices on $\ell_1$ are mapped to $u$, and vertices on $\ell_2$ are mapped to $v$. All other vertices in the region are mapped arbitrarily to $u$ or $v$.}
  \end{subfigure}
  \caption{Using non-intersecting curves to find an embedding from face $F$ to $H$.}
  \label{fig:proof-intuition}
\end{figure}

Given a face $F$ and a set of curves valid with respect to $F$, we can give a stretch-1 retraction from $G$ to $H$.  In essence, the curves partition the graph into regions such that all vertices in a particular region map to one of two end-points of a particular edge of $H$. 
See Figure~\ref{fig:proof-intuition} for an illustration.

Of course, it is not obvious that a valid set of curves exists for a given face, and, if it does, how to compute it.
We show that if the graph has a stretch-1 retraction, then there is some face $F$ with $k$ valid curves, and that we can efficiently compute them.
Our algorithm (Algorithm~\ref{alg:planar-retraction}) iterates over all faces in the graph, in each case finding the maximum number of valid curves it can with respect to that face. 
The number of valid curves we can find is the length of the shortest cycle surrounding $F$.
If the shortest cycle $C$ surrounding $F$ has length $\ell$, then it is impossible to find more than $\ell$ valid curves with respect to $F$:
By the Jordan curve theorem, each curve must intersect $C$, and by the definition, valid curves do not intersect each other and can intersect $C$ only at its vertices.
Our construction of the valid curves shows that this is tight (i.e. we can always find $\ell$ curves).
We show that if a stretch-1 retraction exists, then there is some face for which $\ell = k$.
Algorithm~\ref{alg:planar-retraction} gives an outline of the algorithm.

\begin{algorithm}
  \caption{Outline for finding a stretch-1 retraction, or proving that none exists.}
  \label{alg:planar-retraction}
  \begin{algorithmic}[1]
    \For{inner face $F$ in $G$}
      \State Compute maximum number of valid curves between $F$ and $H$ $p_1,\ldots,p_\ell$ \label{line:compute-curves}
      \If{$\ell = k$} 
        \State Compute stretch-1 retraction from $G$ to $H$ using $p_1,\ldots,p_k$ \label{line:compute-retraction-from-curves}
      \EndIf
    \EndFor
    \State If no retraction was computed, report no stretch-1 retraction exists
  \end{algorithmic}
\end{algorithm}
\subsection{Algorithm and analysis}
This section gives the details of various components of Algorithm~\ref{alg:planar-retraction}, and provides a proof of correctness.
The following is an outline of the rest of the section:
\begin{enumerate}
\item Lemma~\ref{lem:retraction-from-curves} shows how to compute a stretch-1 retraction using the $k$ valid curves in line~\ref{line:compute-retraction-from-curves} of Algorithm~\ref{alg:planar-retraction}.
\item Next, Lemma~\ref{lem:lower-bound} shows that if a stretch-1 retraction exists, there must be some face $F$ in the graph such that the smallest cycle surrounding $F$ has length $k$. 
\item Finally, Lemma~\ref{lem:construction-of-curves} gives a construction of largest set of valid curves for a given face $F$ from line~\ref{line:compute-curves}, and shows that the number of curves computed equals the length of the smallest cycle surrounding $F$.
\end{enumerate}

We begin by showing in Lemma~\ref{lem:paths-partition} a somewhat obvious fact:
A set of valid curves partition $R_H \setminus F$, and each region of the partition contains a single edge of $H$.
We then show in Lemma~\ref{lem:retraction-from-curves} that this partition can be used to produce a stretch-1 embedding.
See Figure~\ref{fig:proof-intuition} for pictorial presentation of these two lemmas.

\begin{lemma}
  \label{lem:paths-partition}
  Let $\{ p_1,\cdots,p_k \}$ be a set of curves that are valid with respect to $F$.
  Let $Z$ denote the set of faces of $H \cup F \cup \bigcup_ip_i$ excluding the outer face and $F$.   Then, each face $f \in Z$ is bordered by exactly 1 edge of $H$, and every vertex of $G \setminus \bigcup_i p_i$ is in a unique face of $Z$.
\end{lemma}
\begin{proof}
  Consider the faces of $H \cup F \cup \bigcup_ip_i$.
  $H$ and $F$ still define faces since the paths $p_i$ fall in $R_H \setminus F$.
  Let $(u,v)$ be an edge of $H$, and consider $X = p_i \cup (u,v) \cup p_j \cup p_F(i,j)$ where $p_i$ is the path containing $u$, $p_j$ is the path containing $v$, and $p_F(i,j)$ is the path on the boundary of $F$ between the vertices where $i$ and $j$ meet $F$ such that $F$ is not contained in $X$.
  If $p_i$ and $p_j$ are both degenerate (i.e., each is empty), then $(u,v) = p_F(i,j)$.
  Otherwise $X$ is a simple cycle.
  We claim that $X$ defines a face.
  In particular, we show that the path $p_F(i,j)$ contains no other vertex of path $p_z$ for all $z \not = i,j$.
  Suppose it does and let $w$ be that vertex.
  Let $w'$ be the vertex adjacent to $w$ on $p_z$.
  Then $w' \in R_H \setminus F$, and so $w' \in X$.
  The other end of path $p_z$, call it vertex $y$, is in $H$, but $y \not = u,v$.
  By the Jordan curve theorem, $p_z \setminus w$ must cross $X$.
  Since the graph is planar, $p_z \setminus w$ must contain a vertex of $F, H, p_i$, or $p_j$.
  Any of these outcomes leads to a contradiction.
\end{proof}

\begin{lemma}
\label{lem:retraction-from-curves}
Given a non-outer face $F$ and a set $\{p_1, p_2, \ldots, p_k\}$ of curves that are valid with respect to $F$, we can construct a stretch-1 retraction from $G$ to $H$ in polynomial time. 
\end{lemma}
\begin{proof}
Let $Z$ be as defined in Lemma~\ref{lem:paths-partition}.  For each
vertex $w$ on $p_i$, map $w$ to the unique vertex $v \in H \cap p_i$.
Otherwise, map $w$ to $u$ or $v$, where $(u,v)$ is the unique edge of
$H$ contained in the same face of $Z$ as $w$.  Fix a face $f$ of
$Z$.  Let $(u,v)$ be the unique edge of $H$ contained in $f$. Any edge
$(x,y)$ contained in $f$ also has $x,y \in f$, and so $x$
and $y$ are each mapped to either $u$ or $v$.  Thus, this retraction
to $H$ has stretch 1.
\end{proof}

As mentioned earlier, we will show that our construction produces $\ell$ valid curves for face $F$, where $\ell$ is the minimum length cycle surrounding $F$.  So we must show that if a stretch-1 retraction exists, there is some $F$ such that every cycle surrounding $F$ has length at least $k$.

\begin{lemma}
\label{lem:lower-bound}
Fix a plane embedding of $G$ where $H$ defines the outer face of the embedding and suppose there is a stretch-1 retraction $G$ to $H$.
Then there exists a non-outer face $F$ such that every cycle surrounding $F$ has length at least $k$.
\end{lemma}

\begin{proof}
We prove a related claim that implies the statement in the lemma.  Fix
some stretch-1 retraction of $G$ to $H$.  Then there exists a
non-outer face $F$ such that for every cycle $C$ in the set of cycles
surrounding $F$, and for each vertex \(v \in H\), there is some vertex
of $C$ mapped to \(v\).  This implies that each of these cycles has
length at least $k$, since the statement says that vertices of $C$ are
mapped to $k$ vertices of $H$.
  
The claim is very similar to Sperner's lemma, and the proof is similar
as well.  Let \(\phi: V(G) \rightarrow V(H)\) denote the retraction.
We associate a score with each cycle \(C\) of the graph: Order the
vertices of the cycle in clockwise order, denoted
\(v_1,v_2,\ldots,v_j,v_{j+1}=v_1\).  Consider the sequence
\(\phi(v_1),\ldots,\phi(v_j),\phi(v_{j+1})\).  Let the score of \(C\)
be 0 to start.  For each pair \(\phi(v_i),\phi(v_{i+1})\), we have:
either \(\phi(v_i) = \phi(v_{i+1})\), or \(\phi(v_i)\) and
\(\phi(v_{i+1})\) are adjacent in $H$.  If \(\phi(v_{i+1})\) is
clockwise of \(\phi(v_i)\) (i.e. if they are in the same order as on
\(C\)), add 1 to the score of \(C\).  If they are in counterclockwise
order, subtract 1.  If they are the same vertex, the score remains the
same.  If \(\phi(v_1),\ldots,\phi(v_j)\) does not contain every vertex
on the outer cycle, the score of \(C\) must be 0, since each edge
along the path \(\phi(v_1),\ldots,\phi(v_{j+1})\) is traversed exactly
the same number of times in each direction.  On the other hand, a
cycle with a non-zero score must have a score that is divisible by
$k$.

Next, we claim that the score of cycle $C$ is the same as the sum of the scores of the cycles defining the faces contained in $C$.
To see this, consider the total contribution to the scores of these cycles from any fixed edge.
If the edge is not in cycle $C$, it is a member of exactly 2 faces contained in $C$, and contributes either 0 to both faces, or \(-1\) to one and \(1\) to the other.
Edges in $C$ are each a member of just one face surrounded by $C$.
Therefore, the score of cycle $C$ is the same as the sum of scores of its surrounded faces.
Since the score of cycle $H$ is $k$, there must be some face $f$ that has non-zero score.

Finally, we show that there is some face with nonzero score such that
every cycle surrounding the face also has nonzero score.  Suppose this
is not the case.  Then, every face with a non-zero score is surrounded
by a cycle with score 0, which implies that the sum of all scores of
faces with non-zero scores is 0.  This is a contradiction, since it
implies that the sum of scores of all internal faces in the graph is
0.
\end{proof}

We complete the section by giving a construction of the largest set of valid curves with respect to some face $F$, and show that the number of curves equals the length of the shortest cycle surrounding $F$.  Our curves will be disjoint paths in a supergraph $G_\Delta(F)$ of $G$.  It is necessary to relate the maximum number of disjoint paths to the length of the shortest cycle surrounding $F$.
The following lemma, proved in Appendix~\ref{app:planar}, establishes this connection.
We believe this lemma should be known, but cannot find it in the relevant literature.

\begin{lemma}
\label{lem:cut-contains-cycle}
Let \(G\) be a triangulated graph.
The graph induced by any minimum $s$-$t$ vertex cut is the shortest simple cycle separating \(s\) and \(t\).
\end{lemma}

If $G$ was already triangulated, we could compute a set of vertex disjoint paths from $F$ to $H$ (note that a set of vertex disjoint paths yields a set of valid curves).
By Menger's theorem and Lemma~\ref{lem:cut-contains-cycle}, we would find $\ell$ paths, where $\ell$ is the shortest cycle surrounding $F$.
$G$ may not be triangulated, so instead we could first triangulate $G$ and then compute the paths.
However, the number of paths we find in this case is the length of the shortest cycle surrounding $F$ in the triangulation of $G$, which may be smaller than $\ell$.
We prevent this from happening by producing a triangulation that adds vertices as well as edges.

\begin{lemma}
\label{lem:construction-of-curves}
  Fix a planar embedding of $G$ with $H$ as the outer face, and let
  $F$ be other face.  Then we can compute $\ell$ valid curves in
  polynomial time, where $\ell$ is the length of the shortest cycle
  surrounding $F$.
\end{lemma}

\begin{proof}
We build a triangulated graph $G_\Delta(F)$ from the planar embedding of $G$.
First, add vertices and edges to every face of $G$, excluding the outer face and $F$.
We do this such that (1) every face except $F$ and the outer face is a triangle, and (2) the distance between any $u,v \in G$ is preserved.
From each face with more than 3 edges, we create one new face that has one fewer edge.
One step of this iterative construction is shown in Figure~\ref{fig:face-triangulation}.

\begin{figure}[t]
  \centering
  \begin{subfigure}[c]{.3\textwidth}
    \centering
    \includegraphics[width=.6\textwidth,page=1]{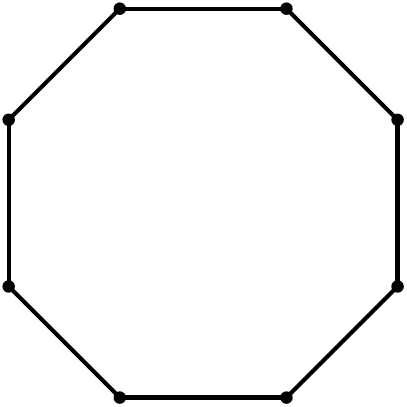}
    \caption{Some face $F'$ with y>3 edges.}
  \end{subfigure}
  \begin{subfigure}[c]{.3\textwidth}
    \centering
    \includegraphics[width=.6\textwidth,page=2]{face-triangulation.pdf}
    \caption{Add a new cycle $C$ with $y-1$ edges inside $F'$ along with connecting edges.}
  \end{subfigure}
  \begin{subfigure}[c]{.3\textwidth}
    \centering
    \includegraphics[width=.6\textwidth,page=3]{face-triangulation.pdf}
    \caption{Add stars in the newly created faces, except the one formed by $C$. Distances between vertices of the original face are preserved.}
  \end{subfigure}
  \caption{Iteratively triangulate faces.}
  \label{fig:face-triangulation}
\end{figure}

Note that distances are preserved inductively, and we make progress by
reducing the size of some face.  The graph we produce has 3 edges
bordering each face, except for the outer face and $F$.  In all, the
number of vertices and edges added to each face of $G$ is polynomial
in the number of edges bordering the face.

Finally, we add vertices $s$ and $t$, and edges from $s$ to each vertex of $F$ and from $t$ to each vertex of $C$.
The resulting graph is triangulated, and we call this graph $G_\Delta(F)$.

At this point, we can find the maximum set of vertex disjoint paths between $s$ and $t$ in $G_\Delta(F)$, by setting vertex capacities to 1 and computing a max flow between $s$ and $t$.
Because we have preserved distances between vertices of $G$ in our construction of $G_\Delta(F)$, the length of the minimum cycle surrounding $F$ must be $\ell$.
Therefore, the number of disjoint paths we find must also be $\ell$.  Finally, we claim that this set of disjoint paths from $F$ to $H$ in $G_\Delta(F)$ is a set of valid curves for $G$.
This is because $G$ is a subgraph of $G_\Delta(F)$, and therefore the criteria for valid curves are still met after removing the vertices and edges of $G_\Delta(F) \setminus G$.
\end{proof}

We conclude by tying together the pieces of the section to show we proved Theorem~\ref{thm:planar}.

\begin{proof}[Proof of Theorem~\ref{thm:planar}]
  Fix a face $F$.
  By Lemma~\ref{lem:cut-contains-cycle}, we determine the set of $\ell$ disjoint paths from $F$ to $H$ where the surrounding minimum cycle is of length $\ell$.
  By Lemma~\ref{lem:lower-bound}, there is a stretch-1 retraction only if there exists a face $F$ whose surrounding min-cycle is of length $k$.
  So if there is no stretch-1 retraction, we find $< k$ disjoint paths for all faces, and our algorithm returns ``no''.
  Otherwise, there exists a face $F$ for which the surrounding min-cycle is of length $k$, and this gives a set of $k$ valid paths.
  Then, by Lemma~\ref{lem:retraction-from-curves}, the retraction that we construct has stretch 1.
\end{proof}


\section{Open problems}
\label{sec:open}
Our work leaves several interesting directions for further research.  First, we would like to determine improved upper and/or lower bounds on the best approximation factor achievable for retracting a general graph to a cycle.  Second, we would like to explore extending our approach for planar graphs (Section~\ref{sec:planar-graph-to-cycle}) and Euclidean metrics (Appendix~\ref{sec:2d_euclidean}) to more general graphs and high-dimensional metrics.  Another open problem is that of finding approximation algorithms for retracting a general guest graph to an arbitrary host graph over a subset of anchor vertices, for which we have presented a hardness result in Appendix~\ref{app:hardness}.    
\newpage

\bibliographystyle{plainurl}
\bibliography{anchored_embedding}

\section*{Appendix}

\appendix

\newcommand{\ora}[1]{\overrightarrow{#1}}
\section{Proofs for Section~\ref{sec:general}}
\label{app:general}

\subsection{Approximation algorithm for general graphs}

\subsection{The Sperner bottleneck}
In this section, we elaborate on the Sperner bottleneck and establish
that natural LP and SDP relaxations for minimum-stretch retraction to
cycles suffer from an $\Omega(\sqrt{n})$ integrality gap on this
instance.  Recall the definition of the Sperner bottleneck instance:
the guest graph $G$ is the $\sqrt{n} \times \sqrt{n}$ grid, and the
host $H$ is the cycle of $G$ formed by the $4\sqrt{n}$ vertices on the
outer boundary of $G$.  In Lemma~\ref{lem:general.lower_bounds}, we
used Sperner's Lemma to show that the optimal stretch achievable for
this instance is $\Omega(\sqrt{n})$.

We now consider natural linear and semi-definite programming
relaxations for the retraction problem, and show that each incurs an
integrality gap of $\Omega(\sqrt{n})$ for the Sperner bottleneck.

\BfPara{Integrality gap of LP relaxation}
A natural LP relaxation for graph retraction is the {\em earthmover
LP}, which has the same constraints as the corresponding LP
extensively studied for the $0$-extension problem, but with a
different objective of taking the maximum stretch, as opposed to the
sum or the average of the
stretches~\cite{karloff+kmr:0-extension,chekuri+knz:label}.  For any
$u \in V$, we have a vector variable $x^u$, which is a probability
distribution over the set of anchors (in $H$).  The stretch of an edge
$(u,v)$ is given by the earthmover distance between $x^u$ and $x^v$,
which can be computed as the minimum cost incurred in sending a unit
flow from $u$ to $v$ in the metric $d_H$.  Here is the earthmover LP
relaxation for the minimum-stretch retraction from $G$ to $H$.

\begin{eqnarray*}
\min & & s  \\
x^u_j - x^v_j + \sum_{i \in H} ((f^{uv})_{ij} - (f^{uv})_{ji}) & = & 0 \;\;\;\;\; \forall (u,v) \in E(G), \forall j \in H\\
\sum_{j \in H} x^u_j & = & 1  \;\;\;\;\; \forall u \in V\\
\sum_{i \in H}\sum_{j \in H}d_H(i,j)(f^{uv})_{ij} & \le & s \;\;\;\;\; \forall (u,v) \in E\\
x, f & \ge & 0
\end{eqnarray*}

We now show that the above LP has objective function value $O(1)$ for
the given instance.  We partition the vertices of the grid into
$\sqrt{n/2}$ groups: for $1 \le i \le \sqrt{n}/2$, the $i$ group $G_i$
denotes the set of $4i$ vertices in the $i$th square at $L_\infty$
distance $1/2 + i-1$ from the geometric center of the grid.  For any
$u$ in $G_i$, we set $x^u$ to be distributed evenly across a segment
in the boundary of length $4(\sqrt{n} - 2i) + 1$.  Thus, for instance,
every vertex in $G_1$ is mapped to a segment of length $4\sqrt{n} -
7$, and every vertex in the boundary (which is $G_{\sqrt{n}/2}$) is
mapped to a segment of length 1 (i.e., to itself).  The segments that
the vertices in a specific group $G_i$ are mapped to are spaced out
evenly across the boundary.

We establish that the earthmover distance between $x^u$ and $x^v$ for
any neighbors $u$ and $v$ of the grid is $O(1)$.  First, suppose $u$
and $v$ are in the same group.  In this case, the minimum-cost flow
that defines the earthmover distance involves sending a flow from
$O(1)$ vertices in one segment across the length of the segment to the
same number of vertices in the other segment; since the distribution
of $x^u$ and $x^v$ across their respective segments is even, this
yields a cost of $O(1)$.  A similar argument holds for edges $(u,v)$
where $u$ and $v$ are in adjacent groups.

\BfPara{Integrality gap of SDP relaxation}
We now introduce an SDP relaxation, partly inspired by similar
relaxations for bandwidth and related
problems~\cite{BLUM200025,manchoso+ye:localize}.  Number the vertices
of $G$ from $1$ to $n$ so that the vertices of $H$ are numbered $1$
through $k$.  Let $v_i$ denote an $n$-dimensional vector representing
vertex $i$. For $1 \le i, j \le k$, let $d(i,j) = \min\{|i-j|, k -
|i-j|\}$.  Our SDP places the $n$ points on a sphere of radius $k$
(second constraint) such that the stretch of every edge is bounded
(third constraint) and the $k$ vertices of $H$ all reside on a cycle
(fourth constraint).  The vertices of $H$ are forced to lie on a cycle
using distance constraints, owing to the following elementary geometry
claim that captures the rigidity of the cycle.  It is used to show
that if we fix the location of 3 points in a cycle on the plane, then
every point on the cycle is uniquely specified by the distances to
these 3 points.  This ensures that the SDP is a valid relaxation of
the problem of retracting to a cycle.
\begin{lemma} \label{k_points_lemma}
Suppose we are given reals $d_{i,j}$ for all $1 \le i \le 3, 1 \le j \le k$ such that there exist points $p_i$ in $R^2$, $1 \le i \le k$ satisfying the property that $d_{i,j}$ is the distance between $p_i$ and $p_j$, $1 \le i \le 3, 1 \le j \le k$.  Then, if $p_1$, $p_2$, and $p_3$ are not collinear, any sequence $p'_i$, $4 \leq i \leq k$ of points in $R^n$, for any $n \geq 2$, which satisfy the same distance properties
must all lie on the same plane as $p_1$, $p_2$, and $p_3$.  And this
point set is congruent to $p_i$'s.
\end{lemma}
\begin{proof} 
Since the distance of each point is defined relative to only three of the points $p_1, p_2, p_3$, we can prove the uniqueness of existence of each $p_i$ for all $4 \le i \le k$ with respect to $p_1, p_2, p_3$ independently. Thus, we only need to prove the following statement: Suppose we have three non-collinear points $A, B, C$ with $n$-dimensional coordinates which lie on a 2-D plane $P$ and there is a point $X$ also in $P$ such that its distance from $A, B, C$ is $d_1, d_2, d_3$. Then any point $Y$ in $n$-dimensional space with distance $d_1, d_2, d_3$ from $A, B, C$ is congruent to $X$.

W.l.o.g, we assume that $A = (a_1, a_2, 0, ..., 0)$, $B = (b_1, b_2, 0, ..., 0)$, $C = (c_1, c_2, 0, ..., 0)$, $X = (x_1, x_2, 0, ..., 0)$ and $Y = (y_1, y_2, y_3, ..., y_n)$. The following equations hold: 
\begin{eqnarray} \label{unique_point_equ}
d_1^2 = (a_1 - x_1)^2 + (a_2 - x_2)^2 = (a_1 - y_1)^2 + (a_2 - y_2)^2 + y_3^2 + ... + y_n^2 \\
d_2^2 = (b_1 - x_1)^2 + (b_2 - x_2)^2 = (b_1 - y_1)^2 + (b_2 - y_2)^2 + y_3^2 + ... + y_n^2 \nonumber \\
d_3^2 = (c_1 - x_1)^2 + (c_2 - x_2)^2  = (c_1 - y_1)^2 + (c_2 - y_2)^2 + y_3^2 + ... + y_n^2 \nonumber
\end{eqnarray}
Then we have: 
\begin{eqnarray*}
x_1(b_1- a_1) + x_2 (b_2 - a_2) & = & y_1(b_1 - a_1) + y_2 (b_2 - a_2)\\
x_1 (c_1 - a_1) + x_2 (c_2 - a_2) & = & y_1 (c_1 - a_1) + y_2 (c_2 - a_2)
\end{eqnarray*}
Since we know coordinates of $A, B, C$, $k_1 =(b_1-a_1), k_2 = (b_2 - a_2), k_3 = (c_1 - a_1), k_4 = (c_2 - a_2)$ are fixed. We thus obtain two lines. 
\begin{eqnarray*}
l_1: (x_1 - y_1) k_1 + (x_2 - y_2) k_2 = 0\\
l_2: (x_1 - y_1) k_3 + (x_2 - y_2) k_4 = 0
\end{eqnarray*}
Note that $l_1$ and $l_2$ are two lines that have at least one common point $(y_1 = x_1, y_2 = x)$. These equations have infinite solutions if the slopes of $l_1$ and $l_2$ are also the same; this would mean that $\frac{-k_1}{k_2} = \frac{-k_3}{k_4}$, but then points $A, B, C$ would be collinear, contradicting  our assumption in the lemma.  This implies that for any point $Y (y_1, y_2, ... , y_n)$ with distance $d_1, d_2, d_3$ from $A, B, C$, $y_1 = x_1$ and $y_2 = x_2$. Using equation \ref{unique_point_equ}, we obtain that $y_3 =  ... y_n = 0$, completeing the proof of the lemma.   
\end{proof}

Here is an SDP LP relaxation for minimum-stretch retraction from $G$
to $H$.
\begin{eqnarray*}
\min & & s \\
v_i \cdot v_j & \ge & 0 \;\;\;\;\; \forall i,j \in \{1, 2, \ldots, n\}\\
|v_i|  & = & k  \;\;\;\;\; \forall i \in \{1, 2, \ldots, n\}\\
|v_i - v_j| & \le & s  \;\;\;\;\; \forall (i,j) \in E(G)\\
|v_i - v_j| & = & 2 k \sin(\pi d(i,j)/k) \;\;\;\;\; \forall i \in \{1,2,3\},j \in \{1, 2, \ldots, k\}
\end{eqnarray*} 
We now establish an $\Omega(\sqrt{n})$ integrality gap for the above
SDP.  The value of the SDP is $O(1)$ since the grid can be embedded
with constant distortion on the surface of a 3-dimensional hemisphere,
with the boundary forming the great circle at the base of the
hemisphere.  The locations of the vertices in the 3-dimensional
hemisphere yield the vectors that form a valid solution to the SDP,
with stretch $s$ being the maximum distortion of the embedding, which
is $O(1)$.

\BfPara{Lower bound on approximation ratio of Algorithm~\ref{alg:grid_embedding}}
We now show that for a variant of the grid instance,
Algorithm~\ref{alg:grid_embedding} incurs stretch $\Omega(\sqrt{n})$
away from that of the optimal.  Let $G$ denote the graph obtained
after removing all interior column edges from a
$\sqrt{n} \times \sqrt{n}$ grid.  We observe that there exists a
retraction of $G$ to the cycle in the boundary with stretch 2.  Each
row of vertices of length $\sqrt{n}$ can be mapped to the shorter of
the two segments of the boundary connecting the end vertices of the
row; this ensures no edge is stretched by more than a factor of 2.  On
the other hand, Algorithm~\ref{alg:grid_embedding} will find the
largest empty square to be square of length 1, and any projection
embedding from the center of such a square will map two neighbors in
the square $\Omega(\sqrt{n})$ distance away from one another.

\BfPara{An alternative approach to overcoming the Sperner bottleneck}
We now present a different linear programming based lower bound for
minimum-stretch retraction that incorporates the topological aspects
of retraction as captured by Sperner's Lemma.  The key idea behind the
lower bound is that if there is a stretch-$s$ retraction $f$ from $G$
to a cycle $H$ of length $k$, then no cycle in $G$ of length less than
$k/s$ must ``loop around'' $H$ in $f$.  Formally, every $\ell$-vertex
cycle $C = v_1 \rightarrow v_2 \rightarrow
v_3 \rightarrow \cdots \rightarrow v_\ell \rightarrow v_1$ corresponds
to a walk $f(v_1) \rightarrow f(v_2) \rightarrow
f(v_3) \rightarrow \cdots \rightarrow f(v_\ell) \rightarrow f(v_1)$
along $H$.  If $\ell < k/s$, then since stretch of $f$ is at most $s$,
the length of this walk is less than $k$, implying that the walk does
not loop around $H$.

We now formalize the above intuition in the following linear program.
Let $\ell$ be an integer.  Fix a direction for the undirected cycle
$H$, and refer to the directed cycle has $\overrightarrow{H}$.  Fix
one direction for each edge $e \in E(G)$, such that edges in $H$ have
all same direction as that of $\overrightarrow{H}$; we use
$\overrightarrow{e}$ to refer to the directed edge.  We also fix a
direction for each cycle $C$, and refer to the directed cycle as
$\overrightarrow{C}$.  Below, we use the notation
$\overrightarrow{e} \parallel \overrightarrow{C}$ to mean that
$\overrightarrow{e}$ is in the cycle $\overrightarrow{C}$ and
$\ora{e}$ has the same direction as $\overrightarrow{C}$, and the
notation $\ora{e} \nparallel \overrightarrow{C}$ to mean that the
reverse of $\ora{e}$ is in $\overrightarrow{C}$.

For any pair of the vertices $u$ and $v$ in
$H$, let $p(u,v)$ denote the length of the unique path from $u$ to
$v$ in $\overrightarrow{H}$.  We define the {\em directed distance}\/
from $u$ to $v$ to be $p(u,v)$, if $p(u,v) \le k/2$; and $p(u,v) - k$,
otherwise.  Note that the directed distance for any pair $(u,v)$ is in
$(-k/2, k/2]$.  Let ${\cal C}_\ell$ denote the collection of all
directed cycles with less than $\ell$ edges.  Let ${\cal LP}_\ell$
represent the following linear program.

\begin{eqnarray*}
   x_{\ora{e}}  & = & 1 \;\;\; \forall \ora{e} \parallel \overrightarrow{H}\\
  \sum_{\ora{e}: \ora{e} \in \overrightarrow{C}, \ora{e} \parallel \overrightarrow{C}} x_{\ora{e}} - \sum_{\ora{e}: \ora{e} \in \overrightarrow{C},  \ora{e} \nparallel \overrightarrow{C}} x_{\ora{e}} & = & 0 \;\;\; \forall \overrightarrow{C} \in {\cal C}_\ell 
\end{eqnarray*}

\begin{lemma} 
\label{lem:cycles_lower_bound}
If there is a stretch-$s$ retraction from $G$ to a cycle $H$ in $G$, then ${\cal LP}_{k/s}$ is feasible.
\end{lemma}
\begin{proof}
Let $f$ be a stretch-$s$ retraction from $G$ to $H$.  For any edge
$\ora{e} = (u,v)$, we set $x_{\ora{e}}$ to be the directed distance from
$f(u)$ to $f(v)$.  Then, for any edge $\ora{e}$ in $H$ along the same
direction as $H$, we have $x_{\ora{e}} = 1$, as required by the linear
program.  Consider any cycle $\overrightarrow{C}$ in ${\cal C}_{k/s}$.
Since the stretch of $f$ is $s$, the total distance in the walk on $H$
induced by $f$ and $\overrightarrow{C}$ is strictly less than $k$.
Since the walk returns to its starting point, it must be the case that
the total directed distance in the walk equals 0.  Therefore,
\[
  \sum_{\ora{e}: \ora{e} \in \overrightarrow{C}, \ora{e} \parallel \overrightarrow{C}} x_{\ora{e}} - \sum_{\ora{e}: \ora{e} \in \overrightarrow{C},  \ora{e} \nparallel \overrightarrow{C}} x_{\ora{e}} = 0 \;\;\; \forall \overrightarrow{C} \in {\cal C}_\ell,
  \]
satisfying the remaining constraint of the linear program, and thus guaranteeing its feasibility.
\end{proof}

By Lemma~\ref{lem:cycles_lower_bound}, the largest $s$ such that
${\cal LP}_{k/s}$ is infeasible is a lower bound on the optimal
stretch achievable.  It is easy to see that for the Sperner bottleneck
instance, the largest $s$ such that ${\cal LP}_{k/s}$ is infeasible is
$\Omega(\sqrt{n})$, hence asymptotically matching the Sperner's Lemma
lower bound.

We now show that for any $\ell$, ${\cal LP}_{\ell}$ can be solved in
polynomial time, by providing a suitable separation oracle.  The
oracle would simply be one of the cycles of length less than $\ell$
for which the corresponding constraint is violated.  We construct
graph $G'$, which has the same set of vertices as $G$ and, for each
edge $e = (u,v)$ in $G$, have an edge between $u$ and $v$ in each
direction.  We set the weight for each directed edge
$\overrightarrow{e}$ in $G'$ to be $x_{\ora{e}}$.  If
$\overleftarrow{e}$ is the reverse of $\ora{e}$, then we set
$x_{\overleftarrow{e}} = -x_{\ora{e}}$.

Let $x$ be a solution to the LP.  Fix an edge $(u,v)$ of $G$.  We
compute a shortest path of less than $\ell$ hops from $v$ to $u$ by
running a shortest path computation in a directed acyclic graph
defined by the hop-expanded (also sometimes referred to as
time-expanded) version of $G'$.  If the shortest path contains a cycle
of non-zero weight, then such a cycle serves as a separating linear
constraint.  Otherwise, if the length of the shortest path from $v$ to
$u$ is not $-x_{(u,v)}$, then the cycle obtained by appending the edge
$(u,v)$ to the shortest path from $v$ to $u$ yields a cycle of
non-zero weight, serving as a separating linear constraint.  We
execute this procedure for all edges $(u,v)$ in $G$; if no separating
linear constraints are found, then $x$ is a feasible solution to the
LP.

\section{Proofs for section~\ref{sec:planar-graph-to-cycle}}
\label{app:planar}
\begin{proof}[Proof of Lemma~\ref{lem:reduction-to-stretch-1}]
Suppose \(G\) can be embedded with stretch \(k\) in \(H\), and let \(f: V(G) \rightarrow V(H)\) be this mapping.  We define an embedding \(f': G_k \rightarrow H\) with stretch 1.  For \((u,v) \in  G\) let \(u = y_0, y_1, \ldots, y_k = v\) be the vertices on the path in \(G_k\) corresponding to edge \((u,v)\) in \(G\).  We show how to embed this path into \(H\). We know that \(d_H(f(u),f(v)) \le k\).
Let \(f(u) = x_0, x_1, \ldots, x_{j} = f(v)\) be the shortest path between \(f(u)\) and \(f(v)\) in \(H\) (note that \(j \le k\)).  Then, the retraction $f'$, defined by setting $f'(y_i)$ to be $x_{\scriptsize \min(i,j)}$, has stretch 1.
\junk{\begin{equation}
f'(y_i) =
  \begin{cases}
    x_i \quad \text{if } i \le j \\
    x_j \quad \text{if } i > j
  \end{cases}
\end{equation}
has stretch 1.}
Conversely, a mapping that produces a stretch-1 retraction of $G_k$ into $H$, when restricted to the vertices in $G$, gives a stretch-$k$ retraction of $G$ into $H$.
\end{proof}

\begin{lemma}
  \label{lem:partial-retraction}
  If there is a stretch-1 retraction $\hat{f}$ from $G$ to subgraph $\hat{G}$ of $G$, then there is a stretch-1 retraction $f$ from $G$ to subgraph $H$ of $\hat{G}$ if and only if there is a stretch-1 retraction $g$ from $\hat{G}$ to $H$.  Furthermore, $f$ can be computed from $\hat{f}$ and $g$ in polynomial time.
\end{lemma}
\begin{proof}
One direction follows immediately: $\hat{G}$ is a subgraph of $G$, so a stretch-1 retraction of $G$ to $H$ implies the same for $\hat{G}$ to $H$.  For the other direction, let $g$ be a stretch-1 retraction from $\hat{G}$ to $H$.  Define $f: G \rightarrow H$ as follows: if $v \in V(\hat{G})$, then $f(v) = \hat{f}(v)$; otherwise, $f(v) = \hat{f}(g(v))$.  Then $f$ is a stretch-1 retraction from $G$ to $H$.  Clearly, $f$ can be computed in polynomial time from $\hat{f}$ and $g$.
\junk{\begin{equation}
  f(v) = 
  \begin{cases}
    \hat{f}(v) \text{ if } v \in V(\hat{G}), \\
    \hat{f}(g(v)) \text{ otherwise.}
  \end{cases}
\end{equation}
}
\end{proof}

\begin{proof}[Proof of Lemma~\ref{lem:2-connected}]
  Suppose $G$ is not 2-vertex connected, and let vertex $v$ be a vertex cut.
  Let $G_1$ and $G_2$ be the disconnected components created after removing $v$ from $G$.
  Since $H$ is a cycle, $H$ is contained completely in either $G_1$ or $G_2$.
  Suppose WLOG that $H \in G_1$.
  Mapping every vertex in $G_2$ to $v$ yields a stretch-1 retraction to $G' = G \setminus G_2$.
  By Lemma~\ref{lem:partial-retraction}, there is a stretch-1 retraction from $G'$ to $H$ if and only if there is a stretch-1 retraction from $G$ to $H$. 
  We can repeat this procedure until we obtain a 2-connected graph.
\end{proof}

\begin{proof}[Proof of Lemma~\ref{lem:cut-contains-cycle}]
  Consider the dual graph $G^*$ for some planar embedding of $G$.
  This graph is constructed by placing a vertex $u^*$ in each face of $G^*$, and adding an edge between $u^*$ and $v^*$ if the faces are adjacent in $G$.
  Note that there is a correspondence between vertices of $G$ and faces of $G^*$, as well as between faces of $G$ and vertices of $G^*$.
  Additionally there is a one-to-one correspondence between edges.
  It is well known that $X \subseteq E$ contains an $s-t$ cut if and only if $X^*$ contains the edges of some simple cycle separating $s$ from $t$.
  We will use this fact to prove our result.

  We need to show that a set of vertices $Y$ in $G$ separates $s$ from $t$ if and only if the subgraph induces by $Y$ contains a cycle separating $s$ from $t$.
  One direction is straightforward: if the induced graph on $Y$ forms a cycle separating $s$ from $t$, then applying the Jordan curve theorem tells us that any path from $s$ to $t$ must cross this cycle (and thus contain a vertex of $Y$).
  To prove the converse, we use the fact given above.
  Let $Y$ be an $s-t$ vertex cut in $G$ and let $E_Y$ be the set of edges in the graph induced by $Y$.
  Then from the fact above, $E^*_Y$ in the dual graph $G^*$ contains a cycle $C^*$ separating $s$ and $t$.
  Because the graph is triangulated, we can show that $E_Y$ also contains a separating cycle:
  edges $e_1^*$ and $e_2^*$ that are adjacent in $E^*_Y$ correspond to edges $e_1$ and $e_2$ falling on the same face of $G$.
  Because each face of $G$ has only 3 edges, $e_1$ and $e_2$ must therefore be adjacent.
  Therefore, $C$ (the edges corresponding to $C^*$) must also be a cycle.
\end{proof}

\section{Retracting points on a $2$-D plane to a uniform cycle}
\label{sec:2d_euclidean}
\junk{
\begin{definition}
Given a {\em guest metric}\/ $d_G$ over a set $V$ of vertices and a {\em host metric}\/ $d_H$ over subset $A \subseteq V$, a mapping $f: V \rightarrow A$ is a retraction of $X$ to $A$ if $f(p) = p$ for all $p \in A$.  For a given retraction $f$ of $V$ to $A$, define the {\em stretch}\/ of a pair $(u,v)$ of points in $X$ to be $d_G(f(u),f(v))/d_H(u,v)$, and the {\em stretch}\/ of $f$ to be the maximum stretch over all pair of points in $V$.  The goal of the {\em minimum metric retraction}\/ problem is to find a retraction of $V$ to $A$ with minimum stretch.
\end{definition}
}
In this section, we consider an Euclidean metric variant of the problem of retracting a graph to a cycle.  Formally, let $V$ be a set of $n$ points in two-dimensional Euclidean plane with a subset $A \subseteq V$ of $k$ {\em anchors}, which are evenly placed on a circle in the plane. For any two points $x$ and $y$ in the plane, let $d(x,y)$ denote the Euclidean distance between $x$ and $y$.  For a retraction $f: V \rightarrow A$, we define the {\em stretch}\/ of a pair $(u,v)$ of points in $X$ to be $d(f(u),f(v))/d(u,v)$, and the {\em stretch}\/ of $f$ to be the maximum stretch over all pair of points in $V$.

A plausible approach to attacking this Euclidean metric retraction problem is to use techniques from Euclidean distance geometry and network localization problems, which involve similar distance-based embeddings.  A popular approach for these problems is using SDP relaxations~\cite{manchoso+ye:localize,alfakih+kw:distance,bavrinok:distance}.  Though there is a large body of literature on distance geometry, graph rigidity, and network localization~\cite{hendrickson:realize,graver+ss:rigidity,aspnes+gy:localize,eren+gwymab:localize,savvides+hs:localize}, the specific objectives being pursued in our graph retraction formulation is different than these problems, as a result of which none of their results seem to yield good results to our Euclidean metric retraction problem.

In this section, we build on our optimal retraction algorithm for planar graphs to develop a constant-factor approximation algorithm for minimum-stretch retraction of $V$ to $A$.  Let the anchors $A$ be labeled $a_0, a_1, \ldots, a_{k-1}$.  Without loss of generality, we assume that the distance between $a_i$ and $a_{i+1 \bmod k}$ is 1.  

We begin our analysis of Algorithm~\ref{alg_eudliean_points_to_cycle} with a straightforward upper bound on the maximum stretch achieved by an optimal retraction for any metric to the uniform cycle.

\begin{lemma}
\label{lem:euclidean.stretch.bound}
For any metric, a stretch of $nk/2$ is achievable for retracting to a uniform cycle.
\end{lemma}
\begin{proof}
Let $G$ denote the weighted complete graph over the set $V$ of $n$ points where the weight of each edge is the distance between the two points in the metric.  We perform a series of contractions in $G$ that finally yield a new graph $\tilde{G}$.  Let $G'$ be an intermediate graph.  The vertices of $G'$ form a partition of $V$: each vertex in $G'$ is a disjoint subset of $V$.  It is convenient to denote each vertex $v$ of $G$ by the singleton set $\{v\}$, so that every vertex of every intermediate graph is a subset of $V$.  Consider the contraction of any edge with length less than $1/n$.  When an
edge $(X,Y)$, where $X$ and $Y$ are disjoint subsets of $V$, is contracted, we replace $X$ and $Y$ and the edge $(X,Y)$ by a single vertex $Z = X \cup Y$ and replace edges as follows: for every edge $(X,W)$ with $W \neq Y$, $Z$
has an edge $(Z,W)$ of the same weight.  Similarly, for every edge
$(Y,W)$ with $W \neq X$, $Z$ has an edge $(Z,W)$ of the same weight.  If there are multiple edges between two vertices, we remove all but the edge with the lowest weight.

We obtain $\tilde{G}$ when no more contractions are possible.
Assuming that the distance between $a_i$ and $a_{i+1 \bmod k}$ is 1, we claim that any two anchors $a_i$ and $a_j$ are members of distinct vertices in $\tilde{G}$. This is because any contraction can decrease distance between two vertices by at most $1/n$, and the maximum number of possible contractions is at most $n-1$ since each contraction decreases the number of vertices by 1.  Now, the
distance between any two vertices in the contracted graph is at most
the distance between the two vertices in the original metric.

The distance between any two vertices $X$ and $Y$ in $\tilde{G}$ is at least $1/n$.  By our construction, it follows that the distance in $G$ between any $x \in X$ and $y \in Y$ is at least $1/n$.  
Now consider any retraction $\phi$ in $G$ that satisfies the following properties: for any $X \in \tilde{G}$, if $X$ contains an anchor $a_i$, $\phi$ maps all vertices in $X$ to $a_i$.  The stretch of $\phi$ is at most
$kn/2$ since the only vertices that are mapped to distinct anchors are at distance at least $1/n$ and any two anchors are at distance at most $k/2$.
\end{proof}

One consequence of Lemma \ref{lem:euclidean.stretch.bound} is that any optimal retraction has to map any two vertices at distance less than $2/(kn)$ to the same anchor. Note that by a series of contractions (at most $n-1$), the distance of any pair of vertices will decrease by no more than $2/k$. 

\begin{figure}[htb]
    \centering
    \begin{subfigure}[t]{\textwidth}
        \centering
        \includegraphics[width=10cm]{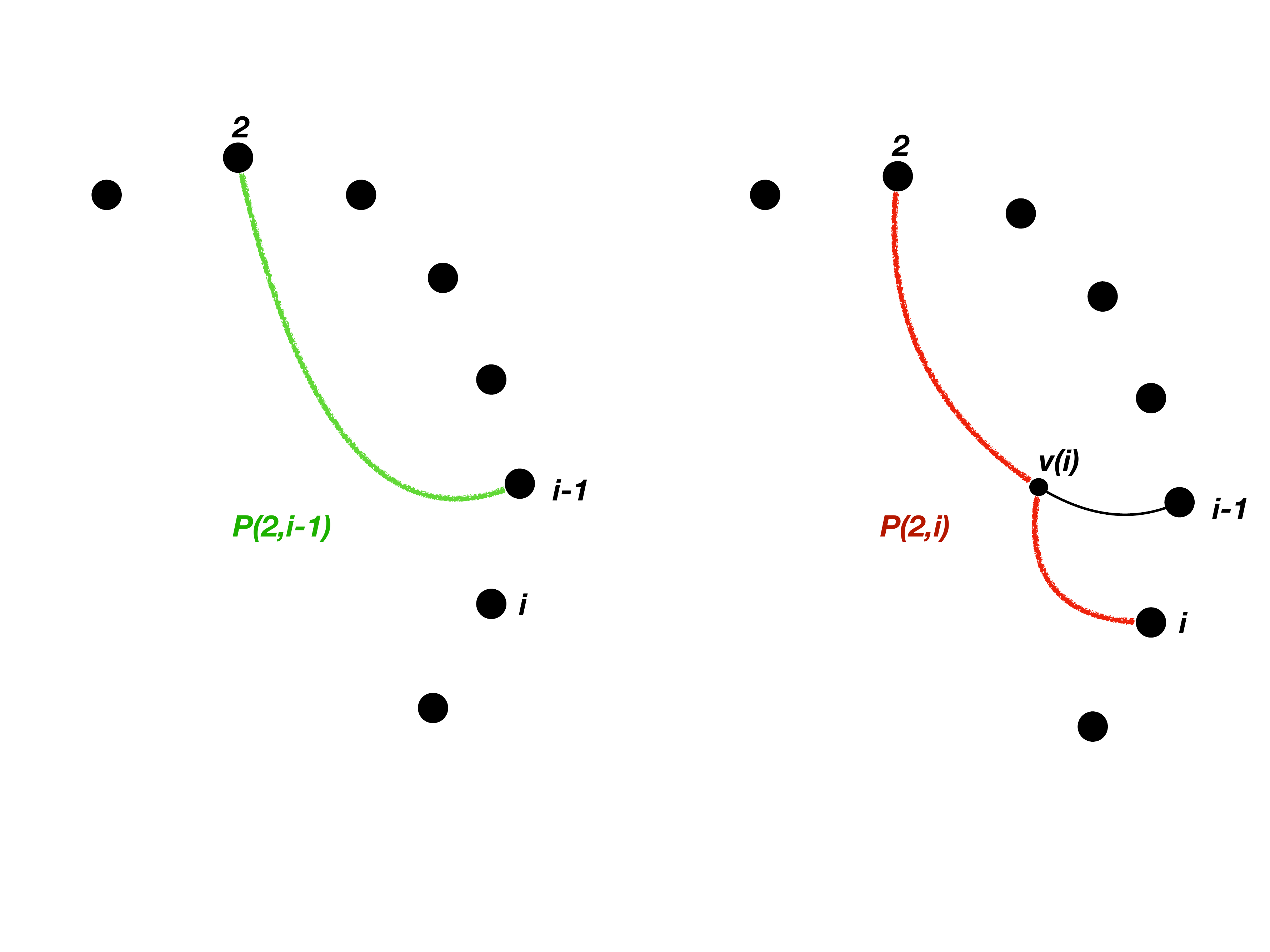}
        \caption{Inductively obtaining path $P(2,i)$ from path $P(2,i-1)$.}
        \label{fig_const_path}
    \end{subfigure}
    ~
    \begin{subfigure}[t]{\textwidth}
        \centering
        \includegraphics[width=10cm]{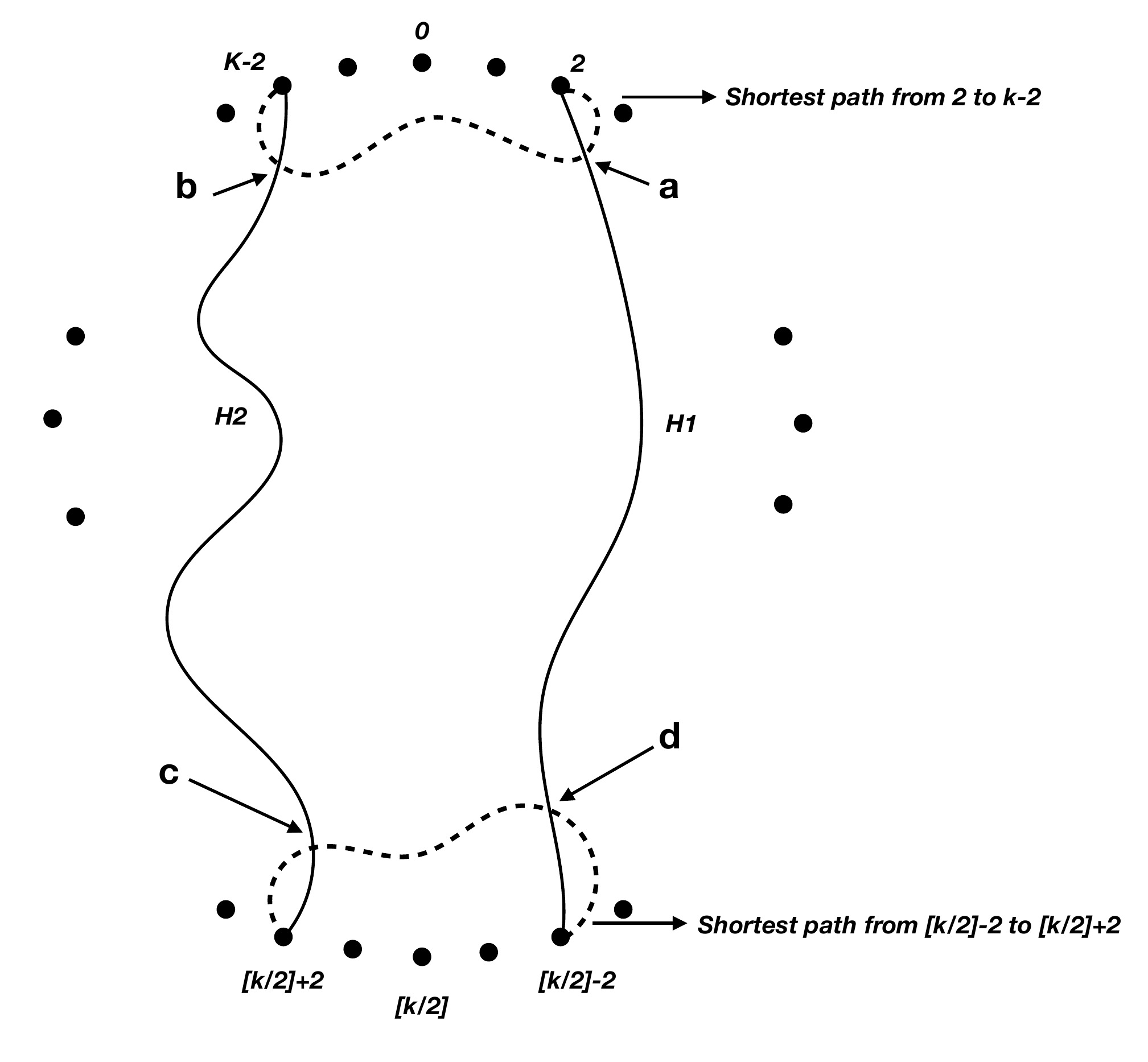}
        \caption{Obtaining cycle $H$ by connecting paths $H_1$ and $H_2$ via paths $P_{ab}$ and $P_{cd}$.}
        \label{fig_const_H}
    \end{subfigure}
    \caption{Construction of $H$.}
\end{figure}

In the following, we analyze the steps of the above algorithm. It is also easy to verify that any sequence of contractions in step~2 and the subsequent actions of the algorithm maintains the planarity of the graph.

\begin{lemma}
\label{lemma_ghat_v}
Any $s$-stretch retraction from $\widehat{G}$ to the anchors is a $\Theta(s)$-stretch retraction from $V$ to the anchors, and vice versa. 
\end{lemma}

\begin{proof}
Using Delaunay triangulation as the planar Euclidean spanner, for any pair of points $u$ and $v$,  $d_{\widehat{G}}(u,v) \leq 2.5 \; d(u,v)$ \cite{Keil:1992:CGA}, where $d_{\widehat{G}}(u,v)$ is the shortest path distance between $u, v$ over graph $\widehat{G}$, and $d(u,v)$ is the Euclidean distance between $u$ and $v$. 

$(\Rightarrow)$ Let $\widehat{\phi}$ be the $s$-stretch retraction for $\widehat{G}$. Use  the same function for retraction from $V$ to anchors. Now for any pair of verticies $u, v \in V$, we have: 
\[\frac{d(\widehat{\phi}(u), \widehat{\phi}(v))}{d(u,v)}  \leq  2.5 \frac{d_{\widehat{G}}(\widehat{\phi}(u), \widehat{\phi}(v))}{d_{\widehat{G}}(u,v)} \leq  2.5 \;  s 
\]
This implies that $\widehat{\phi}$ gives a $\Theta(s)$-stretch retraction from $V$ to anchors. 

$(\Leftarrow)$ Let $\phi$ be the $s$-stretch retraction from $V$ to anchors. Using the same retraction function for retracting $\widehat{G}$ to anchors, for each $(u,v) \in E(\widehat{G})$, we get: 
\[
\frac{d_{\widehat{G}}(\phi(u), \phi(v))}{d_{\widehat{G}}(u,v)}  \leq  2.5\frac{d(\phi(u), \phi(v))}{d(u,v)} 
 \leq  2.5 \; s  
\]

This implies that $\phi$ gives a $\Theta(s)$-stretch retraction of $\widehat{G}$ to the anchors. 
\end{proof}

\begin{lemma}
\label{lemma_ghat_gbar}
Any $s$-stretch retraction from $\widehat{G}$ to anchors is a $\Theta(s)$-stretch retraction from $\bar{G}$, and vice versa. 
\end{lemma}

\begin{proof}
 Similar to construction in lemma \ref{lem:euclidean.stretch.bound}, the vertices of $\bar{G}$ form a partition of $V$.  

($\Rightarrow$)
Let $\widehat{\phi}$ be a $s$-stretch retraction from $\widehat{G}$ to the anchors.  We define the retraction $\bar{\phi}$ from $\bar{G}$ to the anchors as follows; For each $X \in V(\bar{G})$, which also $X \subset V$, $\bar{\phi}(X) = \widehat{\phi}(v)$ where $v \in X$. Note that by lemma \ref{lem:euclidean.stretch.bound} all vertices which are contracted to a single vertex of $\bar{G}$ have to be retracted to the same anchor. Now for each edge $e(X,Y) \in E(\bar{G})$, where $w(X,Y) = \min\limits_{u\in X, v\in Y, (u,v) \in E(\widehat{G})} w(u,v)$, we have the following: 
\[
\frac
{d_{\bar{G}}(\bar{\phi}(X), \bar{\phi}(Y))}
{w(X,Y)} \leq 
\frac{d_{\widehat{G}}(\widehat{\phi}(u), \widehat{\phi}(v))} 
{w(u,v)} \leq s 
\]
where $(u,v)$ is the minimum weight edge remaining between $X$ and $Y$. 

($\Leftarrow$) Now let $\bar{\phi}$ be a $s$-stretch retraction from $\bar{G}$ to anchors. For each $v \in V$, we define $\widehat{\phi} (v) = \bar{\phi}(X)$, where $v \subset X$. For each edge $e(p,q)\in E(\widehat{G})$, if $\widehat{\phi}(p) = \widehat{\phi}(q)$, then stretch of $e$ is zero. Otherwise, let $p\in X$ and $q \in Y$ such that $\widehat{\phi}(p) \neq \widehat{\phi}(q)$ and $X, Y \in V(\bar{G})$. 
\[
\frac{d_{\widehat{G}}(\widehat{\phi}(p),\widehat{\phi}(q))}{w(p,q)}
\leq 
\frac{d_{\bar{G}}(\bar{\phi}(X),\bar{\phi}(Y)) + 2/k }{w(u,v)} \leq  \frac{2 d_{\bar{G}}(\bar{\phi}(X),\bar{\phi}(Y))}{w(X,Y)} \leq 2s
\]
where $(u,v)$ is the minimum weight edge remained between $X$ and $Y$. Since the actual distance of any pair of anchor is at least one, and series of contraction will decrease distances in $\bar{G}$ by at most $2/k$, we obtain that $d_{\bar{G}}(\bar{\phi}(X),\bar{\phi}(Y)) + 2/k \leq 2 d_{\bar{G}}(\bar{\phi}(X),\bar{\phi}(Y))$. 

\end{proof}

\begin{lemma}
\label{const_dist}
For each anchor $i$, there is a vertex $v$ on $H$ such that $d_{G'}(i,v) \leq O(1)$. Also, for each vertex $v$ on $H$, there is an anchor $i$ such that $d_{G'}(i, v) \leq O(1)$.  
\end{lemma}

\begin{proof}
We need to show that path $H_1$ and $H_2$ do not intersect. Since the distance of anchor $2$ to $k-2$ and similarly anchor $\lfloor k/2 \rfloor-2$ to anchor $\lfloor k/2 \rfloor +2$ are $\Theta(1)$, length of path $P_{ab}$ and similarly $P_{cd}$ are $O(1)$. 

($\Rightarrow$) W.l.o.g. assume that $2\leq i \leq \lfloor k/2 \rfloor-2$. First, we prove that among every five consecutive $ v(i+1), v(i+2), v(i+3), v(i+4), v(i+5)$, at least one of them belongs to $H_1$. Assume that there exists a pair $i < j$ such that $j-i > 5$ and and $v(i)$ and $v(j)$ belongs to $H_1$, but for any $i < k < j$, $v(k)$ does not belong to $H_1$. From the definition, $d_{G'}(j, v(j))$ is no more than the distance of anchor $j-1$ to anchor $j$. Thus $d_{G'}(j, v(j)) \leq 2.5$. Also, $v(j)$ is on the path from $v(i)$ to $i$, otherwise $v(i)$ could not be on $H_1$. Now, if we estimate the distance from anchor $i$ to anchor $j$, we have the following: 
\begin{eqnarray*}
d_{G'}(i, j) & \leq & 
 d_{G'}(i, v(j)) +  d_{G'}(v(j),j)
 \leq d_{G'}(i, v(i)) + d_{G'}(v(j), j)  \leq  5   
\end{eqnarray*}

However, since $j-i > 5$, the distance from anchor $i$ to $j$ is at least $6$, which yields a contradiction. The assumption that none of the five consecutive $ v(i+1), v(i+2), v(i+3), v(i+4), v(i+5)$ belongs to $H_1$ is wrong. Also it is easy to verify that distance of any consecutive $v(i), v(i+1)$ is $O(1)$. Using these two arguments above, for any anchor $2\leq i\leq \lfloor k/2\rfloor-2$, at least one of $v(k)$, for $i \leq k \leq i+5$ or $i-5 \leq k \leq i$, belongs to $H_1$, and $d_{G'}(i, v(k))$ is $O(1)$. Vertex $a$ and vertex $d$ are vertices of $H$ which are at distance $O(1)$ from  anchors $0$, $1$, and respectively anchors $\lfloor k/2\rfloor-1$, $\lfloor k/2\rfloor$. The argument for second half of anchors is exactly the same. This completes the proof for the first part of the lemma. 

($\Leftarrow$) Any vertex on path $P_{ab}$, and similarly path $P_{cd}$ are at distance $O(1)$ of anchor $2$ and respectively anchor $\lfloor k/2\rfloor-2$. W.l.o.g. assume an arbitrary vertex $v$ on $H_1$. From the inductive way of constructing $H_1$, there is some $i$ such that for $i \leq j \leq k/2 $, $v$ belongs to $P(2,j)$, and for $j \leq i$, vertex $v$ is not on $P(2,j)$. Thus, $v$ is a vertex on path from anchor $i$ to vertex $v(i)$. The distance from $v$ to anchor $i$ is at most $2.5$, since the distance from $v(i)$ to anchor $i$ is at most $2.5$.  The argument for vertices of $H_2$ is exactly the same. This completes the proof for the second part of lemma. 
\end{proof}

\begin{lemma}
\label{lemma_gbar_gprim}
There exists a $s$-stretch retraction from $\bar{G}$ to anchors if and only if there is a $\Theta(s)$-stretch retraction from $G'$ to $H$. 
\end{lemma}

\begin{proof}
By graph $G'$'s construction, unweighted planar graph $G'$ has graph $\bar{G}$'s vertices plus some extra auxiliary vertices ($V(\bar{G}) \subset V(G'))$. 

($\Rightarrow$) Let $\bar{\phi}$ be a $s$-stretch retraction from $\bar{G}$ to anchors. We define retraction $\phi'$ from unweighted graph $G'$ to $H$ in the following way: For $X \in V(\bar{G})$, $\phi'(X)$ is the closest vertex on $H$ to $\bar{\phi}(X)$. Using lemma \ref{const_dist}, the distance between $\phi'(X)$ and $\bar{\phi}(X)$ is $O(1)$. 
For each pair of vertices $X,Y\in V(\bar{G})$, retract auxiliary vertices of $V(G')-V(\bar{G})$ between  $X$ and $Y$ uniformly on shortest path on $H$ between $\phi'(X)$ and $\phi'(Y)$. Now consider an arbitrary edge $e(P,Q) \in E(G')$, and let $X$ and $Y$ be two vertices of $V(\bar{G})$ that edge $e$ is on the path between $X$ and $Y$ in graph $G'$.  We bound the stretch of edge $e$ as follows. 
\begin{eqnarray*}
\frac{d_{H}(\phi'(P), \phi'(Q))}{1} \leq \frac{d_{H}(\phi'(X),\phi'(Y))}{d_{G'}(X,Y)} \leq \frac{d_{\bar{G}}(\bar{\phi}(X), \bar{\phi}(Y)) + O(1)} {1/2 \cdot d_{\bar{G}}(X,Y)} \leq 2 s 
\end{eqnarray*}

($\Leftarrow$) Let $\phi'$ be a $s$-stretch retraction from $G'$ to cycle $H$. We define $\bar{\phi}$ to be the retraction from $\bar{G}$ to anchors in the following way: 
For $X\in V(\bar{G})$, let $\bar{\phi}(X)$ to be the closest anchor to $\phi'(X)$. Using lemma \ref{const_dist}, the euclidean distance between $\bar{\phi}(X)$ and $\phi'(X)$ is $O(1)$. Now consider an arbitrary edge $e(X,Y)\in E(\bar{G})$.  We bound the stretch of $\bar{\phi}$ for edge $e$ as follows. 
\begin{eqnarray*}
\frac{d_{\bar{G}}(\bar{\phi}(X), \bar{\phi}(Y))}{w(e)} & \leq & \frac{d_{H}(\phi'(X), \phi'(Y)) + O(1)}{d_{G'}(X,Y)} \leq O(s)
\end{eqnarray*}
where $d_H$ denote the distance of vertices over the cycle $H$. 
\end{proof}

\begin{theorem}
The above algorithm yields a retraction from $V$ to the anchors with stretch within a constant factor of the optimal.
\end{theorem}

\begin{proof}
Applying lemmas \ref{lemma_ghat_v}, \ref{lemma_ghat_gbar}, \ref{lemma_gbar_gprim} immediately implies that retraction $\phi$ from set $V$ to anchors has stretch within a constant factor of the optimal. 
\end{proof}

\begin{algorithm}[t]
\caption{Algorithm for retracting points on a 2-D Euclidean plane to a uniform cycle}
\label{alg_eudliean_points_to_cycle}
\begin{algorithmic}
\Input {Points set $V$ and Anchors set $A \subset V$}
\Output {Retraction function $\phi$ from $V$ to $A$}
\State {1. {\bf Construct a Planar Spanner:} Set weighted graph $\widehat{G}$ to be a planar Euclidean distance spanner for the points set $V$, with the weight $w(e)$ of an edge $e = (u,v)$ denoting the Euclidean distance between the two vertices $u$ and $v$.}
\\
\State {2. {\bf Contract very small edges: }Set weighted graph $\bar{G}$ to be the graph obtained from $\widehat{G}$ by following steps. }
\State{ \indent 2.1. Contract all edges of weight less than $\frac{2}{kn}$.}
    \State{ \indent 2.2. Remove self loops, and among parallel edges, remove all but the shortest edge. }
\\   
\State {3. {\bf Convert to unweighted graph: }Construct unweighted graph $G'$ as follows.} 
	
    \State{ \indent  For each remaining edge $e \in E(\bar{G})$:}
	\State{
    \indent \indent If $w(e) \ge k$, we replace it by a new path consisting of $k^2n/2$ edges. }
    \State{
     \indent \indent Else, replace it by a new path consisting of $\lfloor kn w(e)/2 \rfloor$ edges.}

\\
\State {4. {\bf Construct cycle $H$: } Construct cycle $H$ from underlying graph $G'$ as follows.} 

\State{\indent 4.1. Set $P(2,3)$ to be the shortest path from anchor $2$ to anchor $3$.}
 \State{\indent 4.2. For $4 \leq i \leq \lfloor k/2 \rfloor - 2$:} 
 \State{\indent \indent 4.2.1. Set $v(i)$ to be the closest vertex on path $P(2, i-1)$ to anchor $i$.}
 
\State{\indent \indent 4.2.2. Set $P(2,i)$ to be the union of the portion of $P(2, i-1)$ that connects anchor $2$ to $v(i)$ and shortest path from $v(i)$ to $i$ (see figure \ref{fig_const_path}).}
 
\State{\indent 4.3. Set $H_1$ to be $P(2,\lfloor k/2 \rfloor-2)$.}
 
\State{\indent 4.4. Set $P(k-2, k-3)$ to be the shortest path from anchor $k-2$ to anchor $k-3$.}

\State{\indent 4.5. For $k-4 \geq i \geq \lfloor k/2 \rfloor + 2$:} 

\State{\indent \indent 4.5.1. Set $v(i)$ to be the closest vertex on path $P(k-2, i+1)$ to anchor $i$.}

\State{\indent \indent 4.5.2. Set $P(k-2,i )$ to be the union of the portion of $P(k-2, i+1)$ that connects anchor $k-2$ to $v(i)$ and shortest path from $v(i)$ to anchor $i$ (see figure \ref{fig_const_path}).}

\State{\indent 4.6. Set $H_2$ to be $P(k-2,\lfloor k/2 \rfloor+2)$.} 

\State{\indent 4.7. Consider shortest path from anchor $2$ to anchor $k-2$, let $a$ and $b$ respectively be the last and first (respectively) vertex on the path that intersects $H_1$ and $H_2$ (respectively), let $P_{ab}$ be that portion of the shortest path that connects vertex $a$ and $b$.  
Similarly, consider shortest path from anchor $\lfloor k/2\rfloor-2$ to anchor $\lfloor k/2 \rfloor + 2$, let $c$ and $d$ respectively be the last and first (respectively) vertex on the path that intersects $H_2$ and $H_1$ (respectively), let $P_{cd}$ be that portion of the shortest path that connects vertex $c$ and $d$ (see figure \ref{fig_const_H}).
}
 
 \State{\indent 4.8. Set $H$ to be the cycle forms by union of $H_1$, $P_{ab}$, $H_2$, $P_{cd}$.}
 
\\
\State{5. Set $\phi'$ to be the optimal retraction of $G'$ to $H$, obtained by algorithm \ref{alg:planar-retraction}}.  
\\
\State{6. For each vertex $v \in V$, 
set $\phi(v)$ to be the closest anchor to  $\phi'(v)$ in terms of the actual Euclidean distance.}

\Return $\phi$
\end{algorithmic}
\end{algorithm}

\section{Retraction of bounded treewidth graphs}
\label{sec:treewidth}

In this section, we show that minimum retraction of a bounded treewidth graph to any of its subgraph can be obtained optimally and in polynomial time with respect to size of graph and its subgraph.  
Let $G$ be a graph with bounded treewidth $w$, and $H$ be a subgraph of $G$. We call the vertices of $H$ anchors. Let $k = |V(H)|$. 

In the following, we give a dynamic programming based algorithm which outputs a stretch-$1$ embedding of graph $G$ to $H$ if any exists. 

\BfPara{Algorithm for $1$-stretch retraction} Consider a nice tree decomposition $T$ of graph $G$. 
Nice tree decomposition is a tree decomposition which is a rooted binary tree with four types of vertices called Leaf, Introduce, Forget, Join\footnote{There is an algorithm that converts a tree decomposition to a nice tree decomposition with at most $O(nw)$ bags and in time $O(nw^2)$ \cite{Bodlaender:1993:LTA:167088.167161}. A decomposition of this type is used very often to solve a graph optimization problem with dynamic programming technique. }.   

Let $E(X, f)$ be a binary function that takes a vertex (also called bag) of $T$ and a retraction function $f: X \rightarrow H$ as inputs. If it outputs 1, it means the retraction function $f$ over vertices in $X$ gives a stretch at most 1. Take an arbitrary  vertex $R$ of $T$ as a root and solve $E(X,f)$ for vertices of $T$ in a bottom-up manner. Consider following cases:  

\begin{itemize}
\item {\bf Leaf}: Vertex $X$ is a leaf of $T$. Then $E(X, f) = 0$, if for some anchor vertex $a \in X$, $f(a) \neq a$ or for some pair of adjacent vertices $u,v$ of graph $G$ which $u, v\in X$, $dist_H((f(u), f(v)) > 1$, otherwise $E(X,f) = 1$.

\item {\bf Introduce}: Vertex $X$ has only one child $Y$, and $|X-Y| = 1$ and $Y \subset X$. Let $x = X-Y$. In this case, $E(X,f) = E(Y, f') \land C$, where $f'$ is restriction of function $f$ to domain $Y$, and $C$ is $0$ if either $x$ is an anchor which $f(x)\neq x$ or if for some neighbor $y$ of $x$ in $G$ that $y\in Y$, $dist_{H} (f(x), f(y)) > 1$, otherwise it is 1.  

\item {\bf Forget}: Vertex $X$ has only one child $Y$, which  $X \subset Y$ and $|Y-X| = 1$. Then $E(X,f) = \lor_{f'} E(Y, f')$, $f'$ is any function with domain $Y$ such that for all $x \in X$, $f'(x) = f(x)$. 

\item {\bf Join}: Vertex $X$ has two children $Y_1, Y_2$, which $Y_1$ and $Y_2$ are identical to $X$. Then $E(X,f) = E(Y_1, f) \land E(Y_2, f)$. 

\end{itemize}

Since a nice tree decomposition has at most $O(nw)$ nodes, there are at most $O(nwk^w)$ subproblems, and each subproblem can be computed in time $O(n^2)$ using dynamic programming. So the total running time to find (if any exists) a retraction function $f$ that $E(R,f) = 1$ is $O(n^3 w k^w)$. Since the graph has a bounded treewidth, the retraction function can be extracted in polynomial time with respect to size of graph and its subgraph.  

\BfPara{Algorithm for the general case} To find the optimal retraction of a bounded treewidth graph $G$ to its subgraph $H$, we use the algorithm for $1$-stretch retraction as a subroutine in the following way; Let $G_l$ be the graph which is obtained by replacing any edge $e \in E(G)\backslash E(H)$, by a path of size $l$. 
In lemma \ref{equal_tw}, we show that $G_l$ has the same treewidth as $G$.  
Size of $G_l$ is at most $n + ml$, where $m$ is the number of edges of $G$, and $l$ is no more than $k$. Thus the size of $G_l$ is $O(kn^2)$. It is easy to verify that if retraction of $G_l$ to $H$ has stretch-$1$, then retraction of $G$ to $H$ has stretch of at most $l$. Thus, the optimal retraction of a bounded treewidth graph $G$ to its subgraph $H$, can be obtained by finding smallest $1 \leq l \leq k$ such that retraction of $G_l$ to $H$ has stretch $1$ using algorithm for $1$-stretch retraction. The running time of algorithm is $O(n^6 k^{w+4}w)$. 

\begin{lemma} \label{equal_tw} Let $G_l$ be the graph obtained by replacing each edge of $G$ with a path of length $l$. Treewidth of $G_l$ is the same as $G$. 
\end{lemma}
\begin{proof}

In order to prove the lemma, it is enough to show that replacing a single edge $e=(u,v)$ with a path of length $l$ does not change the treewidth. 
Let $tw(G)$ denote treewidth of graph $G$, and let $G'$ be the graph obtained by replacing an arbitrary edge $e=(u,v)$  of $G$ with a path $P: u, u_1, ..., u_l = v$ of length $l$. 

We show that $tw(G) \leq tw(G')$. Consider an arbitrary tree decomposition of $G'$, in any bag that has $u_i$, $1 \leq i \leq l$, replace it with $v$. In this way, the size of no bag has increased. Also, it is easy to verify that the new decomposition is a valid tree decomposition for $G$. It implies that $tw(G) \leq tw(G')$.   

To show that $tw(G') \leq tw(G)$, we use another equivalent definition of treewidth:
The minimum of maximum clique size among all chordal graphs obtained from graph $G$ minus one is the treewidth of the graph $G$. 
Take any chordal graph $K$ obtained from $G$, replace edge $e=(u,v)$ by path $P$. Add edges from $u$ to $u_i$ for $1 \leq i \leq l$. Call this new graph $K'$. It is easy to verify that $K'$ is a valid chordal graph obtained from $G'$, and the size of maximum clique in $K'$ is the same as $K$. This implies that $tw(G') \leq tw(G)$. 
\end{proof}
\section{Hardness of retracting a weighted graph to an arbitrary metric}
\label{app:hardness}
In this section, we show that the problem of retracting a {\em weighted}\/ graph to an arbitrary metric over a given subset of the vertices so as to  minimize the maximum stretch is hard to approximate within a factor of $O((\log n)^{1/4-\epsilon})$ for any $\epsilon > 0$, unless $NP \subseteq DTIME(n^{poly(\log{n})})$ where $n$ is the total number of vertices in the graph and the metric.  Our argument follows the hardness proof for the $0$-extension problem due to \cite{Karloff:2009:EDM:1654887.1654889}.
In fact, the construction of the hard instance is almost identical, and the analysis differs primarily to ensure the argument extends to the maximum stretch objective.  For completeness, we present all the essential details. 

Our proof is by a reduction from Max-3SAT(5) problem.  An instance of MAX-3SAT(5) is a CNF formula $\phi$ with $n$ variables and $5n/3$ clauses; each clause has $3$ variables and each variable participates in $5$ clauses, appearing in each clause at most once. 

We briefly recall Max-3SAT(5) based on the $k$-prover protocol of \cite{Chuzhoy:2006:HML:1272948.1272956}. Let $\phi$ be a Max-3SAT(5) formula. In the protocol, there are $k$ provers $P_1, ..., P_k$, and a verifier. The verifier, for each pair $(i, j)$ of provers, picks randomly and independently a clause $C_{ij}$ and a distinguished variable $x_{ij}$ from clause $C_{ij}$. $x_{ij}$ is sent to $P_i$, $C_{ij}$ is sent to $P_{j}$, and both $C_{ij}$ and $x_{ij}$ are sent to all other provers. The query sent to each prover has $k \choose 2 $ coordinates. Every prover is expected to return an assignment of the variable or all variables. Then, the verifier checks for each pair $(i,j)$ that the answers of all provers are consistent. 

We introduce some notations first. Given formula $\phi$, let $R$ denote the set of random strings used by verifier. For a given $r \in R$ and $1 \leq i \leq k$, let $q_i(r)$ denote the query sent to prover $P_i$; Let $Q_i = \cup_r \{ q_i(r) \}$, and for each $q_i \in Q_i$ let $\mathcal{A}_i(q_i)$ denote set of all possible answers of prover $P_i$ to $q_i$ that satisfy all the clauses in the query. 

Let $\epsilon$ be a fixed arbitrary real number such that $0< \epsilon < 1$. Formula $\phi$ is a Yes-instance if there is an assignment that satisfies all clauses, and it is a No-instance with respect to $\epsilon$ if for any assignment at most $(1-\epsilon)$ fraction of clauses are satisfiable. 

\begin{theorem}(Proved by \cite{Arora:1998:PVH:278298.278306}) 
There is a constant $\epsilon$, $0 < \epsilon < 1$, such that it is NP-hard to distinguish between a Yes-instance and a No-instance of the Max-3SAT(5) problem. 
\end{theorem}

\subsection{Graph construction}
Our construction closely follows that of~\cite{karloff+kmr:0-extension}.  Given formula $\phi$, we construct an edge weighted guest graph $G$ over union of two sets of vertices $V'$ and $T$. We first build graphs $G_{V'}$ and $G_{T}$, then obtain graph $G$ by combining them. We call vertex set $V'$ nonterminals, and vertex set $T$ terminals (or labels).  The host graph $H$ will defined by the set $T$ and an associated metric $d_T$.  In the following we describe each of these vertex sets and graphs. 
\\
{\bf Graph $G_{V'}$}: Set $V'$ consists of two types of nonterminal vertices.
\begin{itemize}
\item For each $r \in R$, there is a constraint nonterminal vertex $v(r)$
\item For each $1 \leq i \leq k$ and $q_i \in Q_i$, there is a query nonterminal $v(i,q_i)$. 
\end{itemize}
For each $r$ and $1 \leq i \leq k$, there is an edge between $v(r)$ and $v(i,q_i(r))$ with length $1/2$ and weight $\sqrt{k}$.\\
{\bf Graph $G_T$}: Set $T$ consists of two types of terminal vertices. 
\begin{itemize}
\item For each $r\in R$, and for every k-tuple $(A_1, ..., A_k)$ of pairwise strongly consistent answers that $A_i \in \mathcal{A}_i(q_i(r))$, for $1 \leq i \leq k$,  there is a constraint terminal vertex $(v(r), (A_1, ..., A_k))$.

\item For each $1 \leq i \leq k$, and $q_i \in Q_i$ and each answer $A_i \in \mathcal{A}_i(q_i)$ to $q_i$, there is a query terminal vertex $(v(i, q_i) , A_i)$. 
\end{itemize}
Each vertex $(v(r), (A_1, A_2, ..., A_k))$ is adjacent to vertex $(v(i,q_i(r)), A_i)$ and the length of the edge is $1/2$ and weight $1$. 
\\
{\bf Graph $G_{(V' \cup T)}$} (or simply $G$):  In addition to edges of $E(G_{V'})$ and $E(G_T)$, 
there are also edges between a nonterminal and a terminal vertex with weight 1 in the following way; There is an edge between $v(r)$ and $(v(r ), (A_1, ..., A_k))$, if $(A_1, ..., A_k)$ are strongly consistent answers which $A_i \in \mathcal{A}_i(q_i(r))$ for each $1 \leq i \leq k$. Also, there is an edge between $v(i,q_i)$ and $(v(i, q_i),A_i)$, if $A_i \in \mathcal{A}_i(q_i)$. 

Now we define three metrics in the following way: 
\begin{itemize}
\item 
{\bf Metric $M$ on $G_{V'}$}: for any $x, x' \in V'$: $M(x,x') = \min \{k, d_{G_{V'}}(x,x')\}$ 
\item
{\bf Metric $\Delta$ on $G_T$}: for any $t,t' \in T$: $  \Delta(t,t') = min \{k,d_{G_T}(t,t')\}$ 
\item
{\bf Metric  $d_T$} on Vertex set $T$: for any $(x,y), (x',y') \in T$,  $d_T((x,y),(x',y')) = \sqrt{k}M(x,x') + \Delta((x,y), (x',y'))$
\end{itemize}
where $d_{G_{V'}}$ and $d_{G_T}$ respectively denote shortest path distance over graph $G_{V'}$ and $G_T$.  

Given an instance $\phi$ of MAX-3SAT(5) with $k$-provers, we have thus constructed an instance seeking the retraction of guest graph $G$ to the host defined by the set $T$ with metric $d_{T}$.  Our analysis distinguishes between the two cases depending on whether $\phi$ is a Yes instance.    
\subsection{Yes instance}
If $\phi$ is a Yes instance, then there is a strategy of the provers that verifier accepts the SAT formula with probability 1. In this case, provers initially agree on a same satisfactory assignment. For each random string $r$ and $i = 1, ..., k$, let $f_i(q_i(r))$ be the answers of prover $P_i$ to query $q_i(r)$ regarding their strategy. Note that it is the case that $f_i(q_i(r))$s are pairwise strongly consistent. Now we define the assignment of nonterminals to terminals in the following way.  For every random string $r$, assign nonterminal $v(r)$ to terminal $(v(r), ( f_1(q_1(r)), ...,  f_k(q_k(r)))$, for every random string $r$ and each $1 \leq i \leq k$, assign the nonterminal $v(i,q_i(r))$ to terminal $v((i,q_i(r)), f_i(q_i(r)))$. 

Now we show an upper bound of $O(k)$ for the maximum weighted stretch of an edge in the Yes instance. 

We look at three cases: 
\begin{itemize}
\item An edge between two nonterminal vertices, a constraint nonterminal $v(r)$ and a query nonterminal $v(i, q_i(r))$; 
Let $a = (v(r), (f_1(q_1(r)), ...,  f_k(q_k(r))) ), b = (v(i, q_i(r)), f_i(q_i(r)))$ be mapped vertices of $v(r)$ and $v(i, q_i(r))$ respectively. Then $d_T(a,b) = \sqrt{k}M(v(r), v(i, q_i(r))) + \Delta(a, b) \leq 1/2(\sqrt{k}+1)$. Also weight of such an edge is $\sqrt[]{k}$, thus the cost in this case would be $O(k)$. 

\item An edge between a constraint nonterminal $v(r)$ and a constraint terminal $b = (v(r), (A_1, ..., A_k))$; 
Then from the assignment $v(r)$ is assigned to $a = (v(r), (f_1(q_1(r)), ...,  f_k(q_k(r))))$. Then stretch of such edge is $d_T(a,b) = \sqrt[]{k} M(v(r), v(r)) + \Delta(a,b) \leq  0 + k = k$. Since the weight of such an edge is $1$, the cost in this case would be O(k). 

\item An edge between a query nonterminal and a query terminal; 
This case is identical to the second case. 
\end{itemize}

\subsection{No instance}
In this case we show a lower bound of $\Omega(k\sqrt{k})$ for cost of any arbitrary assignment. 
Let $f: V' \rightarrow T$ be an arbitrary assignment of nonterminals to terminals. 
For $v \in V'$, define $g(v)$ and $h(v)$ by $f(v) = (g(v), h(v))$. Let $V_1 = \{ v \in V' : M(v, g(v))  \geq \gamma k \} $ for some small constant $0 < \gamma < 1$. 
We consider two cases: 
\begin{itemize}
\item $V_1$ has at least one member. 

Let vertex $u$ be any nonterminal vertex in $V_1$. Vertex $u$ is either a constraint nonterminal of form $v(r)$ or a query nonterminal  of form $v(i,q_i)$. Let $u$ be of form $v(r)$. Then, there is an edge between $u$ and terminal vertex $x = (v(r), (A_1, ... A_k))$. $d_T(f(u), x)$ is the stretch of edge (u,x) by assignment function $f$. 
\begin{eqnarray}
d_T(f(u), x) & = &  d_T( (g(u),h(u)), (v(r), (A_1, ... A_k)))\\
 & = & \sqrt[]{k} M(g(u), u) + \Delta((g(u),h(u)), (v(r), (A_1, ... A_k))) \\
 & \geq & \gamma k \sqrt{k}  
\end{eqnarray}

The case where vertex $u$ is of form $v(i, q_i)$ is identical to this case. 

\item $V_1$ is an empty set.  

In this case, we initially change the assignment from $f=(g,h)$ to $f'(g', h')$ such that for all $v\in V'$, $g'(v) = v$, and also $f'(v)$ is closest to $f(v)$, according to distance $\Delta$ on $T$.

We show that for some nonterminal-nonterminal edge, the unweighted stretch is at least $k$. 
We know that: 
\begin{eqnarray}
\label{nonterminal_sum}
\sum_{r, i} d_T( f'(v(r)), f'(v(i,q_i(r))) ) \geq \sum_{r, i} \Delta( f'(v(r)), f'(v(i,q_i(r))) ) 
\end{eqnarray}

By Proposition 4.4 and Lemma 4.5 in \cite{Chuzhoy:2006:HML:1272948.1272956}, we have that the right-hand side of (\ref{nonterminal_sum}) is at least ${k \choose 2} \frac{\epsilon}{3} |R|$. Since the total number of nonterminal-nonterminal edges is $k|R|$, there exists a nonterminal-nonterminal edge $(u,v)$ that  $\Delta(f'(u), f'(v)) \geq \epsilon' k$.  

Using triangle inequality over metric $\Delta$, we have the following: 
\[
\Delta(f'(u), f'(v)) - \Delta(f(u), f'(u)) - \Delta(f(v), f'(v)) \leq \Delta (f(u), f(v))
\]

In the following, we show that for all $v \in V'$, $ \Delta(f(v), f'(v)) \leq \gamma k$. Note that the structure of graph $G$ implies that if two nonterminals $x,x'$ are adjacent, and $(x,y)$ is terminal vertex, then there is a terminal vertex $(x',y')$ that is adjacent to terminal $(x,y)$ in graph $G$. This implies that if there is a path between two nonterminals $x, x'$ of length $l$, then there is a corresponding path with length at most $l$ between two terminals $(x,y)$ and $(x',y')$. Thus we get the following for any vertex $v \in V'$: 
\begin{eqnarray}
\Delta(f(v), f'(v)) & = & \Delta((g(v), h(v)), (v, h'(v))) \\
& \leq & M(g(v), v)\\
& \leq & \gamma k
\end{eqnarray}

Thus for edge $(u,v)$ we have that: 
\[
(\epsilon'- 2 \gamma) k \leq \Delta (f(u), f(v))
\]

Choosing $\gamma$ small enough, $(\epsilon' - 2 \gamma)$ is a positive constant. Then the weighted stretch of edge $(u,v)$ is $\Omega(k\sqrt{k})$. This completes the proof. 
\end{itemize}

The main result now follows using the same calculations of the parameters as in~\cite{Chuzhoy:2006:HML:1272948.1272956,karloff+kmr:0-extension}.
\begin{theorem}
 For any constant $\delta > 0$, there is no $O((\log n)^{1/4-\delta})$-approximation algorithm for minimizing the maximum weighted stretch problem unless $NP \subseteq DTIME(n^{poly(\log n)})$. 
\end{theorem}
\junk{
\begin{proof}
As in \cite{Chuzhoy:2006:HML:1272948.1272956}, assuming the size of an instance of MAX-3SAT(5) is $n$ , the size of $R$ is at most $(5n)^{k^2}$, and  $|Q_i| \leq (5n)^{k^2} $ for $1 \leq i \leq k$, and for each $q \in Q_i$, the size of answer set to $q$ is at most $7^{k^2}$. Thus the size of query labels is at most $O( k \cdot 7^{k^2} \cdot (5n)^{k^2})$, and the size of  constraint labels is  $O( 7^{k^2} \cdot (5n)^{k^2} )$ since the number of consistent answers is at most $7^{k^2}$, which implies the size of the construction ($N$) is bounded by $ n^{O({k^2})}$. 

Choosing $k = poly(\log n)$, we get that $k = (\log N)^{1/2-\delta}$ for arbitrary small constant $\delta$. As we showed, there is a gap ratio of $\Omega(\sqrt{k})$ between a No instance and a Yes instance.  This implies that there is no $O((\log N)^{1/4-\delta})$-approximation algorithm for any constant $\delta > 0$, for minimizing the maximum weighted stretch problem unless $NP \subseteq DTIME(n^{poly(\log n)})$.  

\end{proof}
}

\clearpage

\end{document}